\documentclass[journal]{IEEEtran}
\usepackage{amsfonts}
\usepackage{graphicx}
\usepackage{color}
\usepackage{amsmath,amsfonts,amssymb,amsthm,epsfig,epstopdf,url,array}
\usepackage{url,textcomp}
\usepackage{authblk}
\usepackage{cite}
\usepackage{verbatim}
\usepackage{subfigure}
\usepackage{caption}
\usepackage{algorithm,algorithmic}
\usepackage{enumerate}

\newcommand{\bs}{\boldsymbol}
\newtheorem{theorem}{Theorem}
\newtheorem{lemma}{Lemma}

\def\BibTeX{{\rm B\kern-.05em{\sc i\kern-.025em b}\kern-.08em
    T\kern-.1667em\lower.7ex\hbox{E}\kern-.125emX}}
\makeatletter

\newcommand{\Rmnum}[1]{\expandafter\@slowromancap\romannumeral #1@}
\makeatother
\begin{document}
\title{Spectral-Efficiency and Energy-Efficiency of Variable-Length XP-HARQ}
\author{Jiahui~Feng,
        Zheng~Shi,
        Yaru~Fu,
        Hong~Wang,
        Guanghua Yang,
       and Shaodan Ma
\thanks{\emph{Corresponding Author: Zheng Shi.}}
\thanks{Jiahui~Feng, Zheng Shi, and Guanghua Yang are with the School of Intelligent Systems Science and Engineering, Jinan University, Zhuhai 519070, China (e-mails: jfeng@hkmu.edu.hk; zhengshi@jnu.edu.cn; ghyang@jnu.edu.cn).}
\thanks{Yaru Fu is with the School of Science and Technology, Hong Kong Metropolitan University, Hong Kong SAR, China (e-mail: yfu@hkmu.edu.hk).}
\thanks{Hong Wang is with the School of Communication and Information Engineering, Nanjing University of Posts and Telecommunications, Nanjing 210003, China (e-mail: wanghong@njupt.edu.cn).}
\thanks{Shaodan Ma is the State Key Laboratory of Internet of Things for Smart City, University of Macau, Macau, China (e-mail: shaodanma@um.edu.mo).}
}
\maketitle
\begin{abstract}
A variable-length cross-packet hybrid automatic repeat request (VL-XP-HARQ) is proposed to boost the spectral efficiency (SE) and the energy efficiency (EE) of communication systems without the knowledge of channel state information at the transmitter (CSIT). The SE is firstly derived in terms of the outage probabilities by using renewal-reward theorem, with which the SE is proved to be upper bounded by the ergodic capacity (EC). The EC is achievable by VL-XP-HARQ if both the number of accumulated information bits and the maximum number of HARQ rounds increase to infinity. Moreover, to enable as well as facilitate the maximization of the SE, the asymptotic outage probability is obtained in a simple form under high signal-to-noise ratio (SNR). With the asymptotic results, the SE can be maximized by properly designing the number of new information bits in each HARQ round while guaranteeing the outage constraint. By applying Due to Dinkelbach's transform, the fractional programming problem can be transformed into a subtraction form, which can be further decomposed into multiple sub-problems through alternating optimization. Based on the finding that the asymptotic outage probability is an increasing and convex function of the number of information bits, the successive convex approximation (SCA) is adopted to relax each sub-problem to a convex problem. Besides, the EE of VL-XP-HARQ is also investigated. It is corroborated that the upper bound of the EE is attainable when the EC is achieved and the transmission power tends to zero. Furthermore, by aiming at maximizing the EE via power allocation while confining the outage probability within a certain constraint, the similar method to the maximization of the SE can be invoked to solve this problem. Finally, plenty of numerical results are presented for verification, wherein variable-length incremental redundancy HARQ (VL-IR-HARQ) is used for comparison.



\end{abstract}

\IEEEpeerreviewmaketitle
\hyphenation{HARQ}
\section{Introduction}\label{sec:int}
\IEEEPARstart Many communication scenarios, such as smart grid, industrial automation, self-driving car, and unmanned aerial vehicles have entailed far more stringent demands on quality of service (QoS) for beyond fifth generation (B5G) and sixth-generation (6G) networks. 
We are therefore obliged to constantly improve reliability, reduce latency, enhance spectral efficiency (SE), and boost energy efficiency (EE) of wireless systems to confront with unprecedented challenges. Towards these goals, hybrid automatic repeat request (HARQ) plays an increasingly vital role for fulfilling reliable transmissions. 
The conventional HARQ can be classified into three types as per different encoding and decoding operations, i.e., Type-I HARQ, chase combining HARQ (CC-HARQ), and incremental redundancy HARQ (IR-HARQ) \cite{ahmed2021hybrid,shen2022energy,zhang2022performance,meng2022analysis,shi2023outage}. 



HARQ schemes can remarkably ameliorate the reception reliability, but at the price of low SE. To address this issue, the enhancement of the SE for various HARQ systems has been extensively investigated in the literature. 
In order to strengthen the SE of HARQ schemes, \cite{nasraoui2022energy} proposed a genetic algorithm based cross-layer design approach that  optimally allocates transmission powers and appropriately chooses modulation coding scheme to maximize the SE of HARQ. Besides, the Q-learning algorithm was applied to design the codeword length for HARQ systems to maximize the SE of HARQ \cite{mueadkhunthod2022design}. In \cite{szczecinski2010variable}, a variable-length IR-HARQ (VL-IR-HARQ) was proposed and found to be superior to the classical fixed-length IR-HARQ in terms of SE. In \cite{8555645}, the transmission rate in each HARQ round was optimally chosen to maximize the SE of VL-IR-HARQ over Beckmann fading channels. 
Apart from the SE, the EE of HARQ has also been received plenty of research interests particularly for future internet-of-things (IoT) applications. To maximize the EE of Type-I HARQ, the bandwidth and the power were properly allocated by considering Rician fading channels in \cite{leturc2019energy}. Furthermore, 
to account for the impact of time correlation among fading channels, the transmission powers and rate were optimized to maximize the EE of three HARQ schemes by capitalizing on Karush–Kuhn–Tucker conditions in \cite{shi2018energy}. In \cite{peng2023power}, deep reinforcement learning was utilized to minimize the long term transmission power of HARQ by dynamically selecting transmission power and coding rate.




In order to enhance SE as well as reduce latency, multi-packet HARQ was conceived to accommodate multiple packets over the same resource block in each HARQ round. So far, there are two coding strategies to realize multi-packet HARQ, including superposition coding (SC) \cite{assimi2009packet,khreis2020multi} and joint coding (JC) \cite{hausl2007hybrid,jabi2017adaptive,jabi2017boost,jabi2018amc}. The SC enabled multi-packet HARQ are essentially based on the combination of power-domain non-orthogonal multiple access (NOMA) and HARQ \cite{mheich2020performance}. Moreover, in \cite{nadeem2021nonorthogonal}, the failed message and new message were superimposed to form a non-orthogonal HARQ transmission strategy. The packet error rate is minimized while ensuring minimum requirement of the SE through the optimization of time-sharing and power-splitting ratios. In addition, reinforcement learning was utilized to provide enhanced throughput for NOMA assisted HARQ \cite{luo2023reinforcement}.
Regarding JC enabled multi-packet HARQ, multiple packets are amalgamated from code domain that brings additional coding gain \cite{el2010multiple}. A cross-packet HARQ (XP-HARQ) was proposed to jointly encode failed message and new message for retransmissions \cite{hausl2007hybrid}.
In \cite{jabi2017boost}, Jabi {\emph{et al.}} implemented XP-HARQ by using turbo codes, and a sub-optimal rate adaptation algorithm was devised to maximize the SE. 
Furthermore, Trillingsgaard {\emph{et al.}} proposed a dynamic programming based rate adaption scheme to maximize the SE of XP-HARQ in \cite{trillingsgaard2017generalized}. Besides, in \cite{liang2018efficient}, the similar method was applied to maximize the SE of XP-HARQ that is implemented by polar codes and backtrack-freezing decoding. 
To take a step further, a joint iterative decoding algorithm with low complexity was developed to improve the SE of polar-coded XP-HARQ \cite{jin2021mppp}. 
Moreover, the polar-coded XP-HARQ was integrated with multilevel coded modulation to further boost the SE in \cite{Yaoyue2023polar}.
Additionally, a special case of XP-HARQ was proposed in \cite{wang2023new}, where the new information is only introduced in the first retransmission round. The simulation results revealed that the proposed HARQ scheme is able to improve the SE. Furthermore, the deep reinforcement learning (DRL) was invoked to maximize the SE of XP-HARQ with the knowledge of outdated channel state information (CSI) \cite{wuda2023deep}. Nevertheless, most of previous works concerned XP-HARQ were based on either simulations or approximations that lacks of meaningful insights as well. To conquer this issue, \cite{feng2022outage} conducted the exact and the asymptotic analyses of outage probability for XP-HARQ from information-theoretical perspective. 
However, the current research on XP-HARQ is still in its fancy. There are still lots of practical issues to be resolved, including low complexity encoding/decoding implementation, accurate performance evaluation, adaptive resource allocation, etc. We thus fill this gap by confining our attention to XP-HARQ in this paper. 


However, the existing works assumed that the codeword lengths of XP-HARQ in all HARQ rounds keep fixed. Inspired by VL-IR-HARQ \cite{szczecinski2010variable,8555645}, this paper is concerned with an evolved version of XP-HARQ that assumes variable codeword length in each HARQ, which has yet to be reported. We refer to such XP-HARQ scheme as variable-length XP-HARQ (VL-XP-HARQ). Undoubtedly, VL-XP-HARQ is anticipated to achieve significant performance enhancement due to new degree of freedom introduced by flexible settings of blocklength.
Besides, VL-XP-HARQ makes it possible to share the same encoder/decoder across different HARQ rounds owing to its merit of adaptive adjustment of coding rate. Furthermore, due to the lacks of analytical results about XP-HARQ, most of the prior works seldom touch on the resource allocation of XP-HARQ, e.g., the maximization of SE or EE.
Therefore, this paper concentrates on the spectral efficiency (SE) and the energy efficiency (EE) of VL-XP-HARQ. The major contributions of this paper are summarized as follows.
\begin{itemize}
    \item The SE of VL-XP-HARQ is derived in terms of the outage probabilities by capitalizing on renewal-reward theorem. In addition, the relationship between SE and ergodic capacity (EC) is studied. It is deduced that the EC is achievable by VL-XP-HARQ when both the number of accumulated information bits and the maximum number of HARQ rounds approach to infinity.
    \item It is intractable to obtain a general expression of the exact outage probability for VL-XP-HARQ, which precludes its further performance evaluation and optimal design. Instead, the asymptotic outage probability is got in a simple form in the high signal-to-noise ratio (SNR). Moreover, the asymptotic outage probability is proved to be an increasing and convex function with respect to the number of information bits.
    \item In order to maximize the SE while maintaining a maximum allowable outage probability, the numbers of transmission bits in different HARQ rounds are optimally designed. However, the fractional form of the objective function and the non-convex feasible set considerably challenge the optimization. To address this issue, the asymptotic outage probability is used to conserve the computational overhead and Dinkelbach's algorithm is adopted to solve the fractional programming problem. By virtue of the convexity of the outage probability, the successive convex approximation (SCA) assisted alternating optimization is invoked to iteratively update the numbers of information bits in all HARQ rounds.
    \item Aside from the SE, the EE of VL-XP-HARQ is also thoroughly investigated. It is found that the EE of VL-XP-HARQ is upper bounded. The upper bound of the EE is reachable when the EC is achieved and the transmission power approaches to zero. Moreover, the similar approach to the maximization of the SE can be applied to solve the maximization of the EE through proper power allocation while allowing for a maximum outage tolerance.
    \item To stress the significance of the VL-XP-HARQ, the variable-length incremental redundancy HARQ (VL-IR-HARQ) is employed for comparison. It is substantiated that the VL-XP-HARQ outperforms the VL-IR-HARQ in terms of the SE and the EE due to the flexibility of the VL-XP-HARQ scheme.
\end{itemize}

The rest of this paper is organized as follows. Section II introduces the system model. The bounds and the maximization of the SE of VL-XP-HARQ are discussed in Section III. Section IV then studies the performance limit and the maximization of the EE of VL-XP-HARQ. All the analytical results are validated in Section V. At last, Section VI concludes this paper.
\section{System Model}\label{sec:sys}
In this section, the protocol and the  transmission model of VL-XP-HARQ are individually introduced.
\subsection{Protocol of VL-XP-HARQ}
As illustrated in Fig. \ref{model:eps}, in the initial HARQ round, the message ${\rm m}_1$ consisting of ${b_{1}}$ information bits is encoded into a codeword ${{\bf  x}_1}$. The codeword ${{\bf  x}_1}$ contains ${N_{1}}$ symbols drawn from a constellation set $\mathcal X $, i.e., ${{\bf  x}_1} = {\Phi }\left[ {\rm m}_1\right] \in {\mathcal{X} ^{{N_1}}}$, where $\Phi [ \cdot ]$ refers to the encoding operation. 
If the receiver succeeds to reconstruct the message, a new cycle of HARQ will be triggered by informing the transmitter with a positive acknowledgement (ACK) message. Otherwise, if the receiver fails to decode the message, a negative acknowledgement (NACK) message will be fed back to the transmitter. In order to substantially reduce the transmission latency and boost the spectral efficiency, the transmitter will combine the failed message ${\rm m}_{1}$ with a new message ${\rm m}_{2}$ containing $b_2$ information bits to construct a longer message ${{\rm m}}_{[2]}=[{\rm m}_{1},{\rm m}_{2}]$ in the subsequent HARQ round. 
The message ${{\rm m}}_{[2]}$ is then encoded as a codeword ${\bf{x}}_{2}$ composing of ${N_{2}}$ symbols, i.e., ${{\bf{x}}_2} = \Phi \left[ {{{\rm m}_{[2]}}} \right] \in {{\mathcal X} ^{{N_2}}}$. At the receiver, the observations of ${\bf{x}}_1$ and ${\bf{x}}_2$ will be concatenated for joint decoding. If the decoding is successful, the receiver will send back an ACK message and a new HARQ cycle will be immediately initiated. Otherwise, a retransmission is requested by feeding back an NACK message. In the next HARQ round, all the failed messages will be combined with some new information bits for retransmission. The cycle of HARQ will be terminated once the maximum number of transmissions, i.e., $K$, is reached or all the information bits are successfully decoded. Different from \cite{feng2022outage} that the fixed length of the codeword in each HARQ round is assumed, we consider a more general case that the length of codewords could vary with HARQ rounds. 


\begin{figure*}[htbp]
\centering
\includegraphics[width=14cm]{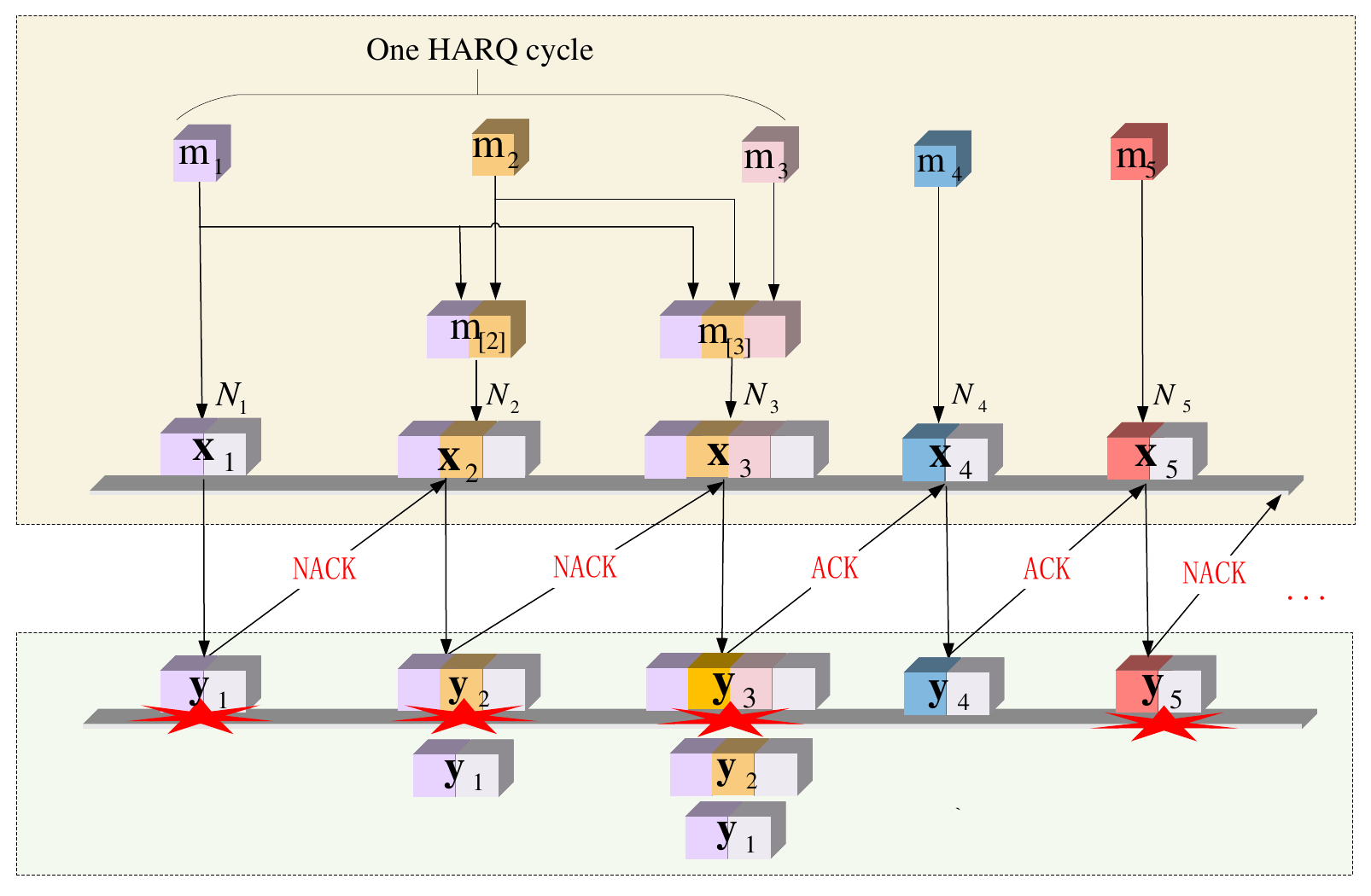}
\caption{An example of the VL-XP-HARQ scheme. 
}\label{model:eps}
\end{figure*}
\subsection{Transmission Model}
In this paper, we assume block Rayleigh fading channels. 
The received signal ${{\bf y}_k}$ in the $k$-th HARQ round is given by
\begin{align}\label{eqn:yk}
{{\bf y}_k} = \sqrt{P_k}{h_k}{{\bf x}_k} + {{\bf z}_k},
\end{align}
where ${P_{k}}$ is the average transmit power in the $k$-th round, ${{\bf z}_k}$ is the complex additive Gaussian white noise (AWGN) with zero mean and variance of ${\sigma ^2}$, ${h_{k}}$ denotes the Rayleigh channel coefficient of the $k$-th HARQ round with normalized average power, i.e., ${\mathbb E}\{ {{{| {{h_k}} |}^2}} \} = 1$, where ${\mathop{\mathbb E}\nolimits} \{\cdot\}   $ stands for the expectation operation. According to \eqref{eqn:yk}, the received signal-to-noise ratio (SNR) in the $k$-th HARQ round is written as
\begin{align}\label{snr}
{\gamma _k} = \frac{{{{\left| {{h_k}} \right|}^2}{P_k}}}{{{\sigma ^2}}}.
\end{align}

On the basis of the above transmission protocol and the channel model, the spectral efficiency and the energy efficiency of VL-XP-HARQ are thoroughly investigated in this paper.

\section{Spectral Efficiency}
In HARQ systems, the spectral efficiency (SE) is usually characterized by the long term average throughput (LTAT). Hence, by letting $t$ count the number of slots, the SE ${\mathcal{T}}_{K}$ measured in bits per second per hertz is given by \cite{caire2001throughput}
\begin{align}\label{eqn:LTAT_infinity0}
{{\cal T}_{K}} = \mathop {\lim }\limits_{t \to \infty } \frac{{{b_\Sigma }\left( t \right)}}{t} = \mathop {\lim }\limits_{t \to \infty } \frac{{\sum\nolimits_{k = 1}^t {b\left( k \right)} }}{{\sum\nolimits_{k = 1}^t {N\left( k \right)} }},
\end{align}
where $b_\Sigma(t)$, ${b\left( k \right)}$, and ${N\left( k \right)}$ denote the successfully received information bits till time slot $t$, the successfully received information bits in time slot $k$, and the length of the codeword delivered in time slot $t$, respectively. According to HARQ protocol, the event that the current cycle of HARQ stops can be recognized as a recurrent event \cite{caire2001throughput}. Since HARQ transmissions can be described as a renewal reward process, \eqref{eqn:LTAT_infinity0} can be derived based on the renewal-reward theorem as
\begin{align}\label{eqn:LTAT_infinity}
{\mathcal{T}_{K}} = \frac{{{\mathop{\mathbb E}\nolimits} \{\mathcal{R}\}}}{{{\mathop{\mathbb E}\nolimits} \{T\}}},
\end{align}
where ${{\mathop{\mathbb E}\nolimits} \{\mathcal {R}\}}$ and ${{\mathop{\mathbb E}\nolimits} \{T\}}$ denote the average number of successfully received information bits and the total average length of the codewords delivered in one HARQ cycle, respectively. 

According to the VL-XP-HARQ protocol, by defining $p_{{suc},k}$ as the probability of the successful event at the $k$-th HARQ round, 
${{\mathop{\mathbb E}\nolimits} \{\mathcal {R}\}}$ can be obtained as
\begin{align}\label{eqn:er}
    {{\mathop{\mathbb E}\nolimits} \{\mathcal {R}\}} &= \sum\nolimits_{k = 1}^K {b_k^\Sigma p_{{suc},k}}\notag\\
   & = \sum\nolimits_{k = 1}^K {b_k^\Sigma ({p_{out,k - 1}} - {p_{out,k}})}\notag\\
   & = \sum\nolimits_{k = 1}^K {b_k ({p_{out,k - 1}} - {p_{out,K}})},
\end{align}
where $b_k^\Sigma = \sum\nolimits_{l=1}^k{b_l}$, the second step holds by using $p_{{suc},k} = {p_{out,k - 1}} - {p_{out,k}}$, and ${p_{out,k }}$ stands for the outage probability after $k$ HARQ rounds. Herein, we stipulate ${p_{out,0}}=1$. Moreover, ${{\mathop{\mathbb E}\nolimits} \{T\}}$ can be obtained as
\begin{align}\label{eqn:et}
    {{\mathop{\mathbb E}\nolimits} \{T\}} 
     &= \sum\nolimits_{k = 1}^K {N_k p_{{out},k-1}}.
\end{align}


By substituting \eqref{eqn:er} and \eqref{eqn:et} into \eqref{eqn:LTAT_infinity}, the SE of VL-XP-HARQ is derived as
\begin{align}\label{eqn:LTAT_FIN}
{{\mathcal T}_{K}} = \frac{\sum\nolimits_{k = 1}^K {b_k ({p_{out,k - 1}} - {p_{out,K}})}}{{\sum\nolimits_{k = 1}^K {{N_k}} {p_{{{out}},k - 1}}}}.
\end{align}
It is easily seen that \eqref{eqn:LTAT_FIN} includes the SE of fixed-length XP-HARQ in \cite[eq. (25)]{jabi2017adaptive} as a special case by setting ${N_1} = {N_2} = \cdots  = {N_K}$. In addition, \eqref{eqn:LTAT_FIN} also incorporates the SE of the conventional HARQ-IR scheme in \cite[eq. (18)]{8555645} as a special case by setting ${b_2} = \cdots = {b_K} =0$, that is, there is no new information bits introduced into retransmissions. 
Clearly, in order to evaluate the SE in \eqref{eqn:LTAT_FIN}, 
it is necessary to derive the outage probability of VL-XP-HARQ.


According to Shannon information theory, it is proved in Appendix \ref{app_out_for} that the outage probability of VL-XP-HARQ can be obtained as
\begin{align}\label{eqn:out_for}
{p_{out,K}} = \Pr \left( {\bigcap\limits_{k = 1}^K {\sum\limits_{l = 1}^k {{I_l}}  < b_k^\Sigma} } \right),
\end{align}
where $\bigcap $ denotes the intersection operation, ${I_k} = {N_k}{\log _2}\left( {1 + {{\gamma _k}}} \right)$, and $b_k^\Sigma = \sum\nolimits_{l = 1}^k {{b_l}} $.

\subsection{Relationship Between SE and EC}
In this section, we investigate the relationship between SE and ergodic capacity (EC). According to the protocol of VL-XP-HARQ and from the information-theoretical perspective, the receiver successfully decodes the message after $\mathcal K$ HARQ rounds if and only if
\begin{align}\label{ineqn:C1}
{I_1} < {b_1} \wedge \cdots  \wedge\sum\nolimits_{k = 1}^{\mathcal K - 1} {{I_k}}  < b_{\mathcal K - 1}^\Sigma  \wedge \sum\nolimits_{k = 1}^{\mathcal K} {{I_k}}  \ge b_{\mathcal K}^\Sigma.
\end{align}
To ease analysis, we define $x_n = \gamma_n$ for $l\in [N_{n-1}+1,\sum_{l=1}^{n} N_l]$.
Thus, \eqref{ineqn:C1} can be rewritten as
\begin{subequations}
\begin{align}\label{ineqn:C2}
&\sum\nolimits_{l = 1}^{\sum\nolimits_{k = 1}^{\kappa } {{N_k}} } {{{\log }_2}\left( {1 + {x_l}} \right)}  < b_{\kappa }^\Sigma ,\,\kappa\in [1,\mathcal K-1] \\
&\sum\nolimits_{l = 1}^{\sum\nolimits_{k = 1}^{\mathcal K} {{N_k}} } {{{\log }_2}\left( {1 + {x_l}} \right)}  \ge b_{\mathcal K}^\Sigma
\end{align}
\end{subequations}

The required number of HARQ rounds in one cycle, i.e., $\mathcal K$, is a random variable whose probability mass function (pmf) is given by
\begin{equation}
    \Pr(\mathcal K = k) = \left\{ {\begin{array}{*{20}{c}}
{{p_{out,k - 1}} - {p_{out,k}}}&{k < K}\\
{{p_{out,K - 1}}}&{k = K}
\end{array}} \right..
\end{equation}
If we assume that all the messages can be almost surely successfully recovered within $K$ transmissions, i.e., $p_{out,K}\to 0$, taking the expectation of \eqref{ineqn:C2} with regard to $x_l$ and $\mathcal K$ leads to
\begin{subequations}
\begin{align}\label{ineqn:EXP}
&{\mathbb E}\left\{\sum\nolimits_{l = 1}^{\sum\nolimits_{k = 1}^{\kappa } {{N_k}} } {{{\log }_2}\left( {1 + {x_l}} \right)}\right\}  < {\mathbb E}\left\{b_{\kappa }^\Sigma \right\},\,\kappa\in [1,\mathcal K-1] \\
&{\mathbb E}\left\{\sum\nolimits_{l = 1}^{\sum\nolimits_{k = 1}^{\mathcal K} {{N_k}} } {{{\log }_2}\left( {1 + {x_l}} \right)}\right\}  \ge {\mathbb E}\left\{b_{\mathcal K}^\Sigma \right\}
\end{align}
\end{subequations}
For simplicity, we assume equal power allocation scheme for VL-XP-HARQ, i.e., $P_1=\cdots=P_K$. Moreover, by ignoring the dependence between $x_l$ and $\mathcal K$, we have
\begin{subequations}\label{EXP}
\begin{align}
\bar C {\mathbb E}\left\{ {\sum\nolimits_{k = 1}^{\mathcal K - 1} {{N_k}} } \right\} &< {\mathbb E}\left\{ b_{\mathcal K - 1}^\Sigma  \right\},\\
\bar C {\mathbb E}\left\{ {\sum\nolimits_{k = 1}^{\mathcal K} {{N_k}} } \right\} &\ge {\mathbb E}\left\{b_{\mathcal K}^\Sigma  \right\},
\end{align}
\end{subequations}
where $\bar C = {\mathbb E}\left\{ {{{\log }_2}\left( {1 + {x_l}} \right)} \right\} $ denotes the ergodic capacity (EC). By noticing that ${\mathbb E}\left\{ {\sum\nolimits_{k = 1}^{\mathcal K} {{N_k}} } \right\} = \mathbb E \{T\}$ represents the average number of symbols within $\mathcal K$ HARQ rounds and ${\mathbb E}\left\{ b_{\mathcal K}^\Sigma  \right\}= \mathbb E\{\mathcal R\} + p_{out,K}b_K^\Sigma\to \mathbb E\{\mathcal R\}$ stands for the average number of delivered information bits, 
\eqref{EXP} can be rewritten as
\begin{subequations}\label{ineqn:ergodic}
\begin{align}
\left( {\mathbb E\{ T \} - \mathbb E\left\{ {{N_{\cal K}}} \right\}} \right)\bar C &< \mathbb E\{\mathcal R\} - \mathbb E\left\{ {{b_{\cal K}}} \right\}, \\
\mathbb E\{ T \}\bar C &\ge \mathbb E\{\mathcal R\},
\end{align}
\end{subequations}
where $\mathbb E\left\{ {{N_{\cal K}}} \right\}$ and $\mathbb E\left\{ {{b_{\cal K}}} \right\}$ refer to the average numbers of symbols and information bits in the last HARQ round, respectively.
On the basis of \eqref{ineqn:ergodic} along with \eqref{eqn:LTAT_infinity}, we arrive at
\begin{align}\label{ineqn:EAC}
{\mathcal T_K} \le \bar C < {\mathcal T_K}\left( {1 - \frac{{\mathbb E\left\{ {{b_{\mathcal K}}} \right\}}}{{\mathbb E\{ {\cal R}\} }}} \right){\left( {1 - \frac{{\mathbb E\left\{ {{N_{\mathcal K}}} \right\}}}{{\mathbb E\{ T\} }}} \right)^{ - 1}}.
\end{align}

It is clear from \eqref{ineqn:EAC} that the SE is less than the EC. In order to achieve the EC, i.e., $\mathcal T_K = \bar C$, according to the squeeze theorem, \eqref{ineqn:EAC} implies that
\begin{align}\label{eq:limit}
    {\frac{{\mathbb E\left\{ {{b_{\mathcal K}}} \right\}}}{{\mathbb E\{ {\cal R}\} }}} = {\frac{{\mathbb E\left\{ {{N_{\mathcal K}}} \right\}}}{{\mathbb E\{ T\} }}} = 0.
\end{align}
Moreover, since \eqref{eq:limit} is based on the assumption of $p_{out,K}\to 0$, such a lossless HARQ requires $K\to \infty$ \cite{larsson2014throughput}. Hence, in order to achieve the EC, as $K\to \infty$, one has
\begin{align}\label{eq:limitfur}
    {\frac{{\mathbb E\left\{ {{b_{\mathcal K}}} \right\}}}{{\mathbb E\{ {\cal R}\} }}} \to 0,  {\frac{{\mathbb E\left\{ {{N_{\mathcal K}}} \right\}}}{{\mathbb E\{ T\} }}} \to 0,~{\rm as }~ K\to \infty.
\end{align}


By noticing that ${\mathbb E\left\{ {{N_{\mathcal K}}} \right\}}$ are bounded with the increase of $K$, i.e., $\min\{ {N_k:k\in [1,K]}\} \le {\mathbb E\left\{ {{N_{\mathcal K}}} \right\}} = \sum\nolimits_{k = 1}^{K - 1} {{N_k}\left( {{p_{out,k - 1}} - {p_{out,k}}} \right)}  + {N_K}{p_{out,K - 1}}\le \max\{N_k:k\in [1,K]\}$, it follows from \eqref{eq:limitfur} that ${\mathbb E\left\{ T \right\}}$ should approach to infinity as $K$ increases, i.e., ${\mathbb E\left\{ T \right\}}\to \infty$. To ensure this condition, we can set $\inf\nolimits_{k\le K-1}\{ p_{out,k}\}>0$ as $K\to \infty$ because of ${{\mathop{\mathbb E}\nolimits} \{T\}} > \inf\nolimits_{k\le K-1}\{ p_{out,k}\}\sum\nolimits_{k = 1}^K {N_k }$, where $\inf$ corresponds to the infimum. Similarly, ${\mathbb E\left\{ {{b_{\mathcal K}}} \right\}}$ is also bounded as $ \min\{b_k:k\in [1,K]\} \le{\mathbb E\left\{ {{b_{\mathcal K}}} \right\}} \le \max\{b_k:k\in [1,K]\}$, which indicates $\mathbb E\{\mathcal R\}\to \infty$ as $K$ increases by capitalizing on \eqref{eq:limitfur}. To guarantee this condition, we can also let $\inf\nolimits_{k\le K-1}\{ p_{out,k}\}>0$, which results in $\mathbb E\{\mathcal R\}\to \infty$ according to \eqref{eqn:er}.

In summary, in order to attain the ergodic capacity, the outage probability of VL-XP-HARQ can be set as $\sup_{K>0}\inf\nolimits_{k\le K-1}\{ p_{out,k}\}>0$ and $p_{out,K}\to 0$ as $K$ tends to infinity, where $\sup $ stands for the supremum. To meet such conditions, we reach the conclusion that $b_K^\Sigma\to \infty$ as $K$ increases, while $b_1,\cdots,b_K$ are bounded.

\subsection{Maximization of SE}\label{sec:max_se}
In the last subsection, it is proved that the ergodic capacity is achievable by VL-XP-HARQ. Towards this end, the number of new information bits introduced into each HARQ round, i.e., ${b_1}, \cdots ,{b_K}$, should be well designed to maximize the SE while maintaining the outage constraint, 
i.e., ${{p_{out,K}} \le \varepsilon }$, where $\varepsilon$ denotes the maximum allowable outage probability. The maximization of the SE is therefore formulated as
\begin{equation}\label{MAX_LTAT}
\begin{array}{*{20}{c}}
{\mathop {\max }\limits_{\bf b} }&{{{\rm{{\cal T}}}_{{{ K}}}}}\\
{\rm s.t.}&{{p_{out,K}} \le \varepsilon },
\end{array}
\end{equation}
where ${\bf b} = (b_1,\cdots,b_K)$. To solve the problem of \eqref{MAX_LTAT}, the closed-form expression of the outage probability ${{p_{out,k}}}$ needs to be obtained. However, it is hardly possible to derive a general expression for ${{p_{out,k}}}$ except for the cases of $k\le 2$, as given below.
\begin{theorem}\label{THEOREM ACCK2}
If $k=1$, the outage probability ${{p_{out,k}}}$ can be calculated as
\begin{equation}\label{eqn:out_k_1_exa}
    {p_{out,1}} = 1 - {e^{ - \left( {{2^{{b_1}/{N_1}}} - 1} \right){\sigma ^2}/{P_1}}}.
\end{equation}
If $k=2$, ${{p_{out,k}}}$ is obtained as \eqref{outage_k2}, as shown at the top of the next page,
\begin{figure*}
\begin{align}\label{outage_k2}
{p_{out,2}} =& {e^{ - \left( {{2^{{b_2}/{N_2}}} - 1} \right){\sigma ^2}/{P_2}}} - {e^{ - \left( {{2^{\left( {{b_1} + {b_2}} \right)/{N_2}}} - 1} \right){\sigma ^2}/{P_2}}}
 + \left( {1 - {e^{ - \left( {{2^{{b_1}/{N_1}}} - 1} \right){\sigma ^2}/{P_1}}}} \right)\left( {1 - {e^{ - \left( {{2^{{b_2}/{N_2}}} - 1} \right){\sigma ^2}/{P_2}}}} \right)\notag\\
& - {e^{{\sigma ^2}/{P_1} + {\sigma ^2}/{P_2}}}H_{1,1}^{1,1}\left( {{2^{\left( {{b_1} + {b_2}} \right)/{N_1}}}\frac{{{\sigma ^2}}}{{{P_1}}}{{\left( {\frac{{{\sigma ^2}}}{{{P_2}}}} \right)}^{{N_2}/{N_1}}}\left| {\begin{array}{*{20}{c}}
{0, - \frac{{{N_2}}}{{{N_1}}},\frac{{{\sigma ^2}}}{{{P_2}}}{2^{{b_2}/{N_2}}}}\\
{0,1,0}
\end{array}} \right.} \right)\notag\\
& + {e^{{\sigma ^2}/{P_1} + {\sigma ^2}/{P_2}}}H_{1,1}^{1,1}\left( {{2^{\left( {{b_1} + {b_2}} \right)/{N_1}}}\frac{{{\sigma ^2}}}{{{P_1}}}{{\left( {\frac{{{\sigma ^2}}}{{{P_2}}}} \right)}^{{N_2}/{N_1}}}\left| {\begin{array}{*{20}{c}}
{0, - \frac{{{N_2}}}{{{N_1}}},\frac{{{\sigma ^2}}}{{{P_2}}}{2^{\left( {{b_1} + {b_2}} \right)/{N_2}}}}\\
{0,1,0}
\end{array}} \right.} \right).
\end{align}
\hrulefill
\end{figure*}
where $H_{p,q}^{m,n}(  \cdot  )$ is the generalized upper incomplete Fox's H function\cite{yilmaz2009productshifted}, $m$, $n$, $p$, and $q$ are integers such that $0 \le m \le q$ and $0 \le n \le p$.
\end{theorem}
\begin{proof}
Please see Appendix \ref{acc_k2}.
\end{proof}

Although the outage probability $p_{out,k}$ in the cases of $k=1, 2$ can be derived in closed-form, the complex form of the exact outage expression still precludes the optimization of the problem \eqref{MAX_LTAT}. To address this issue,
the asymptotic expression of the outage probability in the high SNR regime, i.e., ${P_1/\sigma^2} , \cdots  , {P_K/\sigma^2}\to \infty$, is obtained to ease the optimization of \eqref{MAX_LTAT}, as shown in the following theorem.

\begin{theorem}\label{THEOREM ASY}
The asymptotic outage probabilities in the high SNR regime for $K=1,\,2$ can be obtained as
\begin{equation}\label{eqn:k_1_asy}
    {p_{out,1}}  \simeq \frac{\sigma ^2}{P_1}  ( {{2^{\frac{b_1}{N_1}}} - 1} ),
\end{equation}
\begin{align}\label{eqn:asy_2}
{p_{out,2}} & \simeq  \frac{{{\sigma ^4}}}{{{P_1}{P_2}}}( {\frac{{{N_2}}}{{{N_2} - {N_1}}}( {{2^{\frac{{{b_1}}}{{{N_1}}} + \frac{{{b_2}}}{{{N_2}}}}} - {2^{\frac{{{b_1} + {b_2}}}{{{N_2}}}}}} ) - {2^{\frac{{{b_1}}}{{{N_1}}}}} + 1} ),
\end{align}
where $N_1\ne N_2$. Moreover, the asymptotic outage probability ${p_{out,K}}$ for $K>2$ is given by 
\begin{align}\label{eqn:outage_asy_mul1}
{p_{out,K}} \simeq \prod\limits_{k = 1}^K {\frac{{{\sigma ^2}}}{{{N_k}{P_k}}}} {{h_{K,1}}({x_0})},
\end{align}
where $x_0=1$ and ${\hbar _{K,k}}\left( {{x_{k - 1}}} \right)$ can be calculated as
\begin{align}\label{eqn:out_h_def}
{\hbar_{K,k}}({x_{k - 1}}) 
& = {\left( { - 1} \right)^{K - k + 1}} \prod\limits_{i = 0}^{K-k} {{N_{K-k}}} {x_{k - 1}}^{\frac{1}{{{N_k}}}}  \notag\\
&+ \sum\limits_{i = 0}^{K - k} {{c_{k,i}}{x_{k - 1}}^{\frac{1}{{{N_k}}} - \frac{1}{{{N_{k + i}}}}}},
\end{align}
where ${c_{k,i}}$ can be recursively obtained as
\begin{equation}
    {c_{k,i}} =  - \frac{{{N_{k + i}}{N_k}}}{{{N_{k + i}} - {N_k}}}{c_{k + 1,i - 1}},i \in [1, K - k],
\end{equation}
\begin{align}
    {c_{k,0}}& = {\left( { - 1} \right)^{K - k}}\prod\nolimits_{i = 0}^{K - k } {{N_{K - i}}}{2^{{{b_k^\Sigma }}/{{{N_k}}}}}   \notag
    \\
    &+\sum\nolimits_{i = 1}^{K - k} {{c_{k + 1,i-1 }}\frac{{{N_{k + i}}{N_k}}}{{{N_{k + i}} - {N_k}}}{2^{{{b_k^\Sigma }}/{{{N_k}}} - {{b_k^\Sigma }}/{{{N_{k + i}}}}}}},
\end{align}
and the values of $N_1,\cdots,N_K$ are not identical. In particular, ${\hbar_{K,K}}({x_{K - 1}}) = {N_K}( {{2^{{{b_K^\Sigma }}/{{{N_K}}}}} - {x_{K - 1}}^{{1}/{{{N_K}}}}})$ and $ {c_{K,0}} = {N_K}{2^{{{b_K^\Sigma }}/{{{N_K}}}}}$.
\begin{proof}
Please see Appendix \ref{asy}.
\end{proof}
\end{theorem}
In the following, the simple form of the asymptotic outage probability in Theorem \ref{THEOREM ASY} is used to replace the exact outage probability $p_{out,k}$ in \eqref{MAX_LTAT} to reduce the computational complexity of the optimization. It is worth mentioning that the optimal solution to this relax problem will approach to the true optimal solution at high SNR. This is due to the fact that the asymptotic outage probability becomes more accurate as the SNR increases.

By noticing that problem \eqref{MAX_LTAT} is a fractional programming problem, the Dinkelbach's transform can be adopted \cite{shen2018fractional}. Thus the optimization problem \eqref{MAX_LTAT} can be converted into
\begin{equation}\label{LTAT_MAX_FUN1}
\begin{array}{*{20}{c}}
{\mathop {\max }\limits_{\mathbf b} }&{\tilde {\mathcal R}\left( {{\bf{b}}} \right) - \alpha {\tilde  T}\left( {{\bf{b}}} \right)}\\
{\rm s.t.}&{{\tilde p_{out,K}}\left( {{\bf{b}}} \right) \le \varepsilon }
\end{array},
\end{equation}
where $\tilde{\mathcal R}( {\bf{b}} ) = \sum\nolimits_{k = 1}^K {b_k}( {\tilde p_{out,k - 1}}( {\bf{b}} ) - {\tilde p_{out,K}}( {\bf{b}} ) ) $, $\tilde T\left( {\bf{b}} \right) = \sum\nolimits_{k = 1}^K {{N_k}{\tilde p_{out,k - 1}}\left( {\bf{b}} \right)} $, ${\tilde p_{out,k}}(\mathbf b)$ corresponds to the asymptotic outage probability of VL-XP-HARQ with the numbers of newly introduced information bits set as $\bf b$ and ${\tilde p_{out,k}}(\mathbf b)$ can be evaluated through Theorem \ref{THEOREM ASY}. Moreover, $\alpha$ is the introduced auxiliary variable, which can be iteratively updated as
\begin{align}\label{eta_n_R}
\alpha^{(n+1)} = \frac{{\tilde{\mathcal R}\left( {{\bf{b}}^{\left( {{{n}}} \right)}} \right)}}{{{\tilde T}\left( {{\bf{b}}^{\left( {{{n}}} \right)}} \right)}},
\end{align}
where $n$ the iteration index and we denote by ${\bf b}^{(0)}$ the initial value of $\bf b$ for the iteration algorithm. Unfortunately, the problem of \eqref{LTAT_MAX_FUN1} is still non-convex. To solve this problem, the alternating optimization \cite{J.C.Bezdek2022} is invoked. 
In particular, the numbers of information bits for different HARQ rounds are optimized one by one. Hence, \eqref{LTAT_MAX_FUN1} is decomposed into $K$ single-variable maximization sub-problems. In each inner iteration, the optimization problem can be casted as
\begin{equation}\label{LTAT_MAX_FUN1}
\begin{array}{*{20}{l}}
{\mathop {\max }\limits_{b_r} }&{\tilde{\mathcal R}( {{\bf{b}}_{{r}}^{\left( {{{n}} + {{1}}} \right)}} ) - \alpha^{(n + 1)}{\tilde T}( {{\bf{b}}_{{r}}^{\left( {{{n}} + {{1}}} \right)}} )}\\
{\rm s.t.}&{{\tilde p_{out,K}}( {{\bf{b}}_{{r}}^{\left( {{{n}} + {{1}}} \right)}} ) \le \varepsilon }
\end{array},
\end{equation}
where $r \in \left[ {1,K} \right]$ and 
${\bf{b}}_{r}^{\left( {{{n}} + {{1}}} \right)} = ( b_1^{( {n + 1} )}, \cdots , b_{r - 1}^{( {n + 1} )}, {b_r},b_{r + 1}^{( n )}, \cdots , b_K^{( n )} )$.
However, it is still tremendously challenging to solve \eqref{LTAT_MAX_FUN1} that is an integer programming problem. Nevertheless, by relaxing the feasible region of $b_r$ to a real number, it is found that the outage probability $p_{out,k}$ is a convex function with respect to $b_r$ given the fixed values of $b_1,\cdots,b_{r-1},b_{r+1},\cdots, b_K$, as proved in the following theorem.
\begin{theorem}\label{THEOREM_III}
The asymptotic outage probability ${{\tilde p_{out,K}}}({{\bf{b}}_{r}^{\left( {{{n+1}}} \right)}})$ is an increasing and convex function of ${b_{r}}\in \mathbb R^+$ provided that the numbers of information bits in other HARQ rounds, i.e., $b_1^{( {n + 1} )}, \cdots , b_{r - 1}^{( {n + 1} )}, b_{r + 1}^{( n )}, \cdots , b_K^{( n )} $.
 \end{theorem}
 \begin{proof}
Please see Appendix \ref{outage_convex_proof}.
\end{proof}

With the finding in Theorem \ref{THEOREM_III}, it is obvious that the feasible region of the problem \eqref{LTAT_MAX_FUN1} is convex because of the convexity of ${{\tilde p_{out,K}}}({{\bf{b}}_{r}^{\left( {{{n+1}}} \right)}})$ with respect to $b_r$. Moreover, since ${\tilde T}( {{\bf{b}}_{{r}}^{( {{{n}} + {{1}}} )}} )=\sum\nolimits_{k = 1}^K {{N_k}{\tilde p_{out,k - 1}}( {{\bf{b}}_{{r}}^{( {{{n}} + {{1}}})}} )}$ is a summation of the convex functions, ${\tilde T}( {{\bf{b}}_{{r}}^{( {{{n}} + {{1}}} )}} )$ is also a convex function. However, $\tilde{\mathcal R}\left( {{\bf{b}}_{{r}}^{\left( {{{n}} + {{1}}} \right)}} \right) $ is not convex. To address this issue, successive convex approximation (SCA) can be used to relax \eqref{LTAT_MAX_FUN1} as a convex optimization problem. More specifically, by relying on the first-order Taylor approximation,
 $\tilde {\mathcal R}( {{\bf{b}}_{r}^{\left( {{{n}} + {{1}}} \right)}} )$ is approximated as 
\begin{align}\label{app_G_R}
  &{\tilde {\mathcal R}}( {{\bf{b}}_{r}^{( {{{n}} + {{1}}} )}} ) \approx
  \sum\limits_{k = 1}^r {{b_k}{p_{out,k - 1}}( {{\bf{\tilde b}}_r^{( {n + 1} )}} )}  - \sum\limits_{k = 1}^K {{b_k}} {p_{out,K}}( {{\bf{\tilde b}}_r^{( {n + 1} )}} )\notag\\
  & + \sum\limits_{k = r + 1}^K {{b_k}( {{p_{out,k - 1}}( {{\bf{\tilde b}}_r^{( n+1 )}} ) + {{p}_{out,k - 1}'}( {{\bf{\tilde b}}_r^{( n+1 )}} )( {{b_r} - b_r^{( n )}} )} )}\notag\\
  &\triangleq {\hat {\mathcal R}}( {{\bf{b}}_{r}^{( {{{n}} + {{1}}} )}} )
\end{align}
where $\tilde {\bf{b}}_{r}^{\left( {{{n}} + {{1}}} \right)} = ( b_1^{( {n + 1} )}, \cdots , b_{r - 1}^{( {n + 1} )}, {b_r^{( {n } )}},b_{r + 1}^{( n )}, \cdots , b_K^{( n )} )$ and ${{p}_{out,k - 1}'}( {{\bf{\tilde b}}_r^{( n+1 )}} )$ denotes the first-order partial derivative with respect to $b_r$. By substituting \eqref{app_G_R} into \eqref{LTAT_MAX_FUN1} leads to
\begin{equation}\label{LTAT_MAX_FUN3}
\begin{array}{*{20}{c}}
{\mathop {\max }\limits_{b_r} }&{{\hat {\mathcal R}}( {{\bf{b}}_{r}^{( {{{n}} + {{1}}} )}} ) - \alpha^{(n + 1)}{\tilde T}( {{\bf{b}}_{r}^{\left( {{{n}} + {{1}}} \right)}} )}\\
{\rm s.t.}&{{p_{out,K}}( {{\bf{b}}_{r}^{( {{{n}} + {{1}}} )}} ) \le \varepsilon }.
\end{array}
\end{equation}
Since the objective function of problem \eqref{LTAT_MAX_FUN3} is a concave function with respect to $b_r$ and the feasible region is also a convex set, the sub-problem \eqref{LTAT_MAX_FUN3} can be solved with the classical convex optimization tools. To sum up, the Dinkelbach's algorithm is used to update $\alpha$ in the outer iterations until the difference between $\alpha^{(n+1)}$ and $\alpha^{(n)}$ is negligible or the maximum allowable number of iterations is reached. In the inner iterations, the SCA assisted alternating optimization algorithm is used to update $\bf b$. 
\section{Energy Efficiency}
Apart from the SE, the energy efficiency (EE) is another important performance metric in wireless communications particularly for internet of things (IoT). The EE is defined as the amount of successfully conveyed information bits per unit energy \cite{8319455}. Specifically, 
on the basis of the renewal-reward theorem \cite{zorzi1996use}, the EE ${\eta _K}$ can be written as the ratio of the average number of recovered information bits to the total average consumed energy, i.e.,
\begin{align}\label{eqn:EE}
{\eta _K} = \frac{{\mathbb E}\{\mathcal R\}}{{\mathbb E}\{\mathcal E\}}= \frac{{\sum\nolimits_{k = 1}^K {{b_k}\left( {{p_{out,k - 1}} - {p_{out,K}}} \right)} }}{{\sum\nolimits_{k = 1}^K {{N_k}{P_k}{p_{out,k - 1}}} }},
\end{align}
where ${{\mathbb E}\{\mathcal E\}} = {\sum\nolimits_{k = 1}^K {{N_k}{P_k}{p_{out,k - 1}}} }$ quantifies the average consumed energy in one HARQ cycle. In what follows, we first derive the upper bound of the EE that offers a useful guideline for designing VL-XP-HARQ systems. The transmission powers $P_1,\cdots,P_K$ are then properly designed to maximize the EE while ensuring an outage constraint.
\subsection{Upper Bound of EE}
For simplicity, we assume equal power allocation scheme, i.e., ${P_1} =  \cdots  = {P_K} \triangleq  P$. By using the definition of the EE in \eqref{eqn:EE}, the energy efficiency can be rewritten as
\begin{align}\label{eqn:EEsame}
{\eta _K} = \frac{{\sum\nolimits_{k = 1}^K {{b_k}\left( {{p_{out,k - 1}} - {p_{out,K}}} \right)} }}{{P\sum\nolimits_{k = 1}^K {{N_k}{p_{out,k - 1}}} }} = \frac{{{{\cal T}_{ K}}}}{{P}}.
\end{align}
As proved in \eqref{ineqn:EAC}, the SE ${{{\cal T}_{ K}}}$ is not greater than the EC $\bar C = {\mathbb{E}}\{{{{\log }_2}( {1 + {{\left| h_1 \right|}^2} P/\sigma^2} )} \}$, that is, ${{\cal T}_{ K}} \le \bar C$. Accordingly, the EE is upper bounded as
\begin{align}\label{eqn:EEsame}
{\eta _K} \le \frac{{\mathbb{E}}\{{{{\log }_2}( {1 + {{\left| h_1 \right|}^2} P/\sigma^2} )} \}}{{P}},
\end{align}
where the equality holds if and only if the VL-XP-HARQ achieves the ergodic capacity.
In addition, it is readily proved that the upper bound ${\mathbb E}\{ {{{\log }_2}( {1 + {{\left| h_1 \right|}^2}P/\sigma^2} )} \}/P$ is a decreasing function of ${P}$.
The maximum value of the upper bound is achieved at $P=0$, i.e., ${\mathbb E}\{ {{{\log }_2}( {1 + {{\left| h_1 \right|}^2}P/\sigma^2} )} \}/P < {\mathbb{E}}\{{\left| h \right|^2}\}/ (\sigma^2 \ln 2)$. 
Hence, the maximum EE is bounded as
\begin{align}\label{bound:EE}
 {\eta _K} \le \frac{{{{\mathbb{E}}\{ {{{\left| h _1\right|}^2}} \}}}}{{\sigma^2\ln 2}}.
\end{align}
\subsection{Maximization of EE}
In this section, we focus on the maximization of the EE.
More specifically, the transmission powers in each HARQ round are optimized to maximize the EE while maintaining an outage constraint (i.e., ${p_{out,K}} \le \varepsilon $). The maximization of the EE is therefore formulated as
\begin{equation}\label{MAX_EE}
\begin{array}{*{20}{c}}
{\mathop {\max }\limits_{{P_1}, \cdots ,{P_K}} }&{{\eta _K}}\\
{\rm s.t.}&{{p_{out,K}} \le \varepsilon }
\end{array},
\end{equation}
where ${\bf P}=(P_1,\cdots,P_K)$. Similarly, the asymptotic outage probability in Theorem 2 is used to replace the true outage probability in \eqref{MAX_EE} facilitate the optimization. Based on the finding that \eqref{MAX_EE} takes the similar form as the maximization of the SE in \eqref{MAX_LTAT}, \eqref{MAX_EE} is thus rewritten by applying Dinkelbach's transform as
\begin{equation}\label{EE_MAX_FUN1}
\begin{array}{*{20}{c}}
{\mathop {\max }\limits_{{\bf{P}}} }&{\bar{\mathcal R}\left( {\bf{P}} \right) - \beta \bar {\mathcal E}\left( {\bf{P}} \right)}\\
{{\rm{s}}.{\rm{t}}.}&{{\bar p_{out,K}}\left( {\bf{P}} \right) \le \varepsilon }
\end{array},
\end{equation}
where $\bar{\mathcal R}( {\bf{P}} ) = \sum\nolimits_{k = 1}^K {b_k}( {\bar p_{out,k - 1}}( {\bf{P}} ) - {\bar p_{out,K}}( {\bf{P}} ) ) $, $\bar{\mathcal E}\left( {\bf{P}} \right) = \sum\nolimits_{k = 1}^K {{N_k}P_k{\bar p_{out,k - 1}}\left( {\bf{P}} \right)} $, ${\bar p_{out,k}}(\mathbf P)$ refers to the asymptotic outage probability of VL-XP-HARQ with the transmission powers in different HARQ rounds set as $\bf P$ and ${\bar p_{out,k}}(\mathbf b)$ can be evaluated through Theorem \ref{THEOREM ASY}. Moreover, in the outer iterations, the Dinkelbach's algorithm is used to update the auxiliary variable $\beta$ as
\begin{align}\label{eta_n_R}
\beta^{(n+1)} = \frac{{\bar{\mathcal R}\left( {{\bf{P}}^{\left( {{{n}}} \right)}} \right)}}{{{\bar{\mathcal E}}\left( {{\bf{P}}^{\left( {{{n}}} \right)}} \right)}},
\end{align}
where $n$ the iteration index. 



Similarly, the SCA aided alternating optimization algorithm is used to update $\bf P$ in the inner iterations. Thus, \eqref{MAX_EE} can be divided into $K$ sub-problems and the first-order Taylor expansion can be leveraged to approximate $\bar{\mathcal R}\left( {\bf{P}} \right)$. To proceed with the optimization, we define ${{\bf{P}}_{r}^{{{(n + 1)}}}} = ( {P_1^{\left( {n + 1} \right)}, \cdots ,P_{r - 1}^{\left( {n + 1} \right)},{P_r},R_{r + 1}^{\left( n \right)}, \cdots ,P_K^{\left( n \right)}} )$. Thereupon, each sub-problem reads as
\begin{equation}\label{EE_MAX_FUN2}
\begin{array}{*{20}{c}}
{\mathop {\max }\limits_{{P_r}} }&{\breve{\mathcal R}\left( {{\bf{P}}_r^{(n + 1)}} \right) - {\beta ^{\left( {n + 1} \right)}}\bar{\mathcal E}\left( {{\bf{P}}_r^{(n + 1)}} \right)}\\
{{\rm{s}}.{\rm{t}}.}&{{\bar p_{out,K}}\left( {{\bf{P}}_r^{(n + 1)}} \right) \le \varepsilon }
\end{array},
\end{equation}
where $\breve{\mathcal R}( {{\bf{P}}_r^{(n + 1)}} )$ is given by \eqref{app_G}, as shown at the top of the next page, where $\tilde {\bf{P}}_r^{(n + 1)} = (P_1^{\left( {n + 1} \right)}, \cdots ,P_{r - 1}^{\left( {n + 1} \right)},P_r^{\left( n \right)},R_{r + 1}^{\left( n \right)}, \cdots ,P_K^{\left( n \right)})$ and ${{\bar p}_{out,k - 1}'}( {\tilde{\bf{ P}}_r^{( n+1 )}} )$ denotes the first-order partial derivative with respect to $P_r$.
\begin{figure*}
\begin{align}\label{app_G}
\breve{\mathcal R}\left( {{\bf{P}}_r^{(n + 1)}} \right)\approx &
\sum\nolimits_{k = 1}^r {{b_k}{\bar p_{out,k - 1}}\left( {{\bf{P}}_r^{(n+1)}} \right)}  - \sum\nolimits_{k = 1}^K {{b_k}{\bar p_{out,K}}\left( {{\bf{P}}_r^{(n + 1)}} \right)} \notag\\
 &+ \sum\nolimits_{k = r + 1}^K {{b_k}\left( {{\bar p_{out,k - 1}}\left( {\tilde{\bf{ P}}_r^{(n+1)}} \right) + {\bar p_{out,k - 1}^\prime} \left( {\tilde {\bf{P}}_r^{(n+1)}} \right)\left( {{P_r} - P_r^{(n)}} \right)} \right)},
\end{align}
\hrulefill
\end{figure*}

\section{Numerical Results}
In this section, numerical results are presented for verification. Unless otherwise indicted, we assume equal power allocation for VL-XP-HARQ, i.e., ${P_1} =  \cdots  = {P_K} = P$, and $\sigma^2=1$. The average transmit SNR is defined as $\bar \gamma = P/\sigma^2$. For notational convenience, we denote by ${{\bf {b}}_K}=(b_1,\cdots,b_K)$ the vector of the numbers of new information bits introduced in $K$ HARQ rounds, and denote by ${{\bf {N}}_K}=(N_1,\cdots,N_K)$ the vector of the numbers of symbols in $K$ HARQ rounds.

\subsection{Verification}
In Fig. \ref{sig_out:eps}, the outage probability ${P_{out,2}}$ is plotted against the average transmit SNR $\gamma$ for $K=2$, wherein the labels ``VL-XP-SIM'', ``VL-XP-ANA'', and ``VL-XP-ASY'' represent the simulated, the analytical and the asymptotic results of VL-XP-HARQ, respectively. Moreover, VL-IR-HARQ (labeled as ``VL-IR'') is used for benchmarking purpose \cite{8555645}, wherein the incremental information bits are sent, i.e., $b_2=\cdots=b_K=0$ bits. It is revealed that the analytical and the simulated results are in perfect agreement and the analytical results approach to the asymptotic ones under high SNR. The observations corroborate the validity of our analysis. In addition, it can be seen that the asymptotic outage curves are parallel to each other for different ${{\bf {b}}_2}$, which is due to the same diversity order from Theorem \ref{THEOREM ASY}. Furthermore, it is easily found that increasing the number of information bits will degrade the outage performance. Therefore, it is found that VL-IR-HARQ surpasses VL-XP-HARQ in terms of the outage probability. This is essentially due to the fact that more information bits are delivered, consequently leads to the deterioration of the reliability.
Moreover, Fig. \ref{mul_out:eps} depicts the relationship between the outage probability ${P_{out,K}}$ and the average transmit SNR $\bar \gamma$ for $K=3$. It is observed that the simulated results approach to the asymptotic ones under high SNR, which justifies our analytical results. 
\begin{figure}
\centering
\includegraphics[width=3.5in]{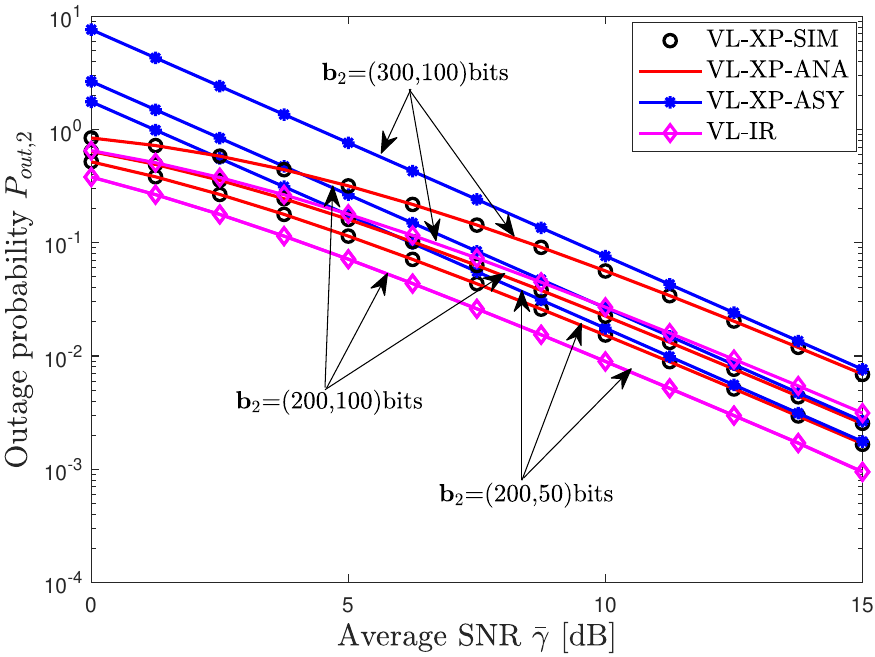}
\caption{The outage probability versus the average transmit SNR $\bar \gamma$ for $K=2$ and ${{\bf {N}}_2}=(100,200)$~symbols.}\label{sig_out:eps}
\end{figure}

\begin{figure}
\centering
\includegraphics[width=3.5in]{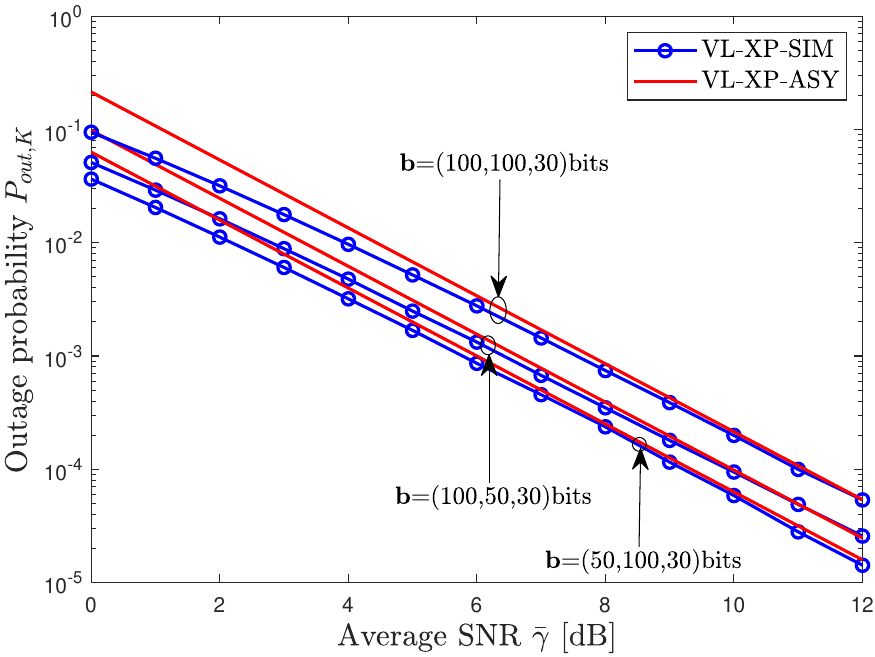}
\caption{The outage probability versus the average transmit SNR $\bar \gamma$ for $K=3$ and ${{\bf {N}}_3}=(100,200,250)$~symbols.}\label{mul_out:eps}
\end{figure}

\subsection{SE}
Fig. \ref{ANA_SIM_LTAT:eps} investigates the SE of VL-XP-HARQ and VL-IR-HARQ against the average transmit SNR for $K=2$. In Fig. \ref{ANA_SIM_LTAT:eps}, the analytical results are in line with the simulated ones, which verify the correctness of our analysis. Besides, it is observed that VL-XP-HARQ performs better than VL-IR-HARQ at low-to-medium SNR in terms of spectral efficiency. In addition, delivering more information bits can improve the SE of VL-XP-HARQ, albeit at the price of increasing its outage probability, as observed in Fig. \ref{sig_out:eps}.
\begin{figure}
\centering
\includegraphics[width=3.5in]{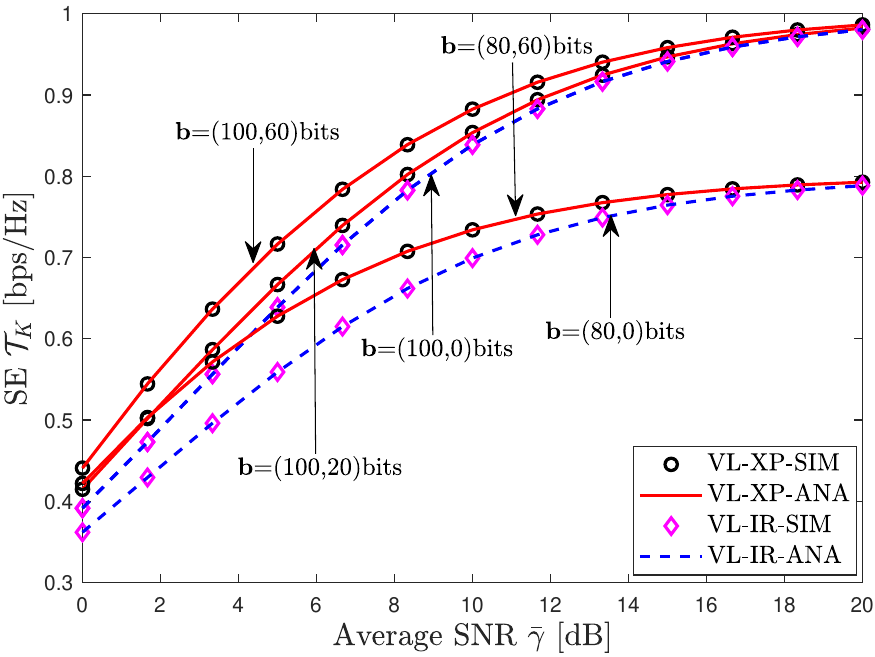}
\caption{The SE versus the average transmit SNR $\bar \gamma$ for $K=2$ and ${{\bf {N}}_2}=(100,200)$~symbols.}\label{ANA_SIM_LTAT:eps}
\end{figure}

In Fig. \ref{EC_VL_LTAT:eps}, the EC and the SE are plotted against the number of the information bits in the initial HARQ round, i.e., ${b_1}$. As $K$ and $b_{1}$ increase, the maximum SE for both VL-XP-HARQ and VL-IR-HARQ tends to the EC. This tendency is consistent with our analysis. In addition, it can be seen that the maximum SE of VL-XP-HARQ is greater than that of VL-IR-HARQ for different $K$. This shows the superiority of VL-XP-HARQ in terms of the SE.
\begin{figure}
\centering
\includegraphics[width=3.5in]{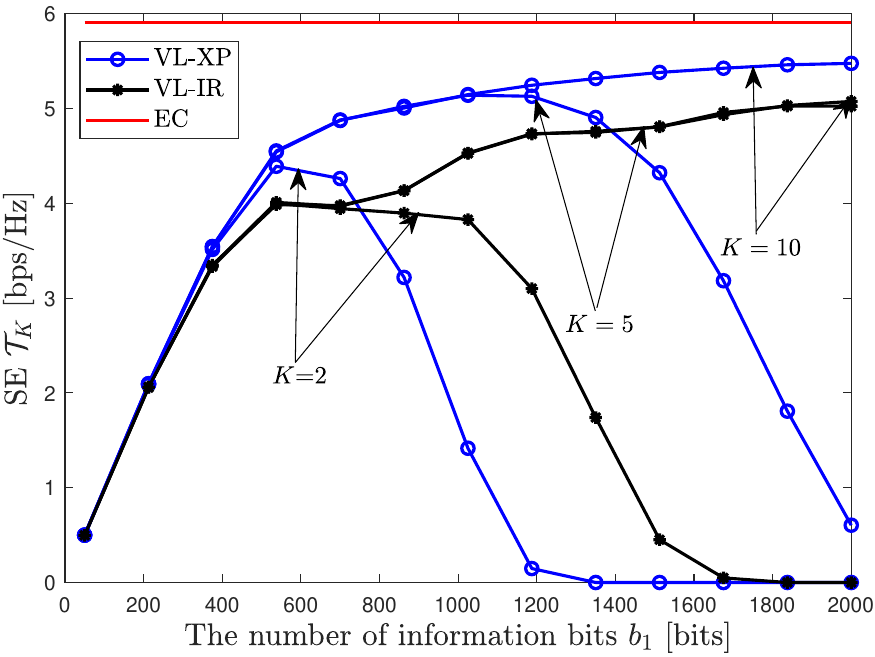}
\caption{The SE versus the number of information bits in initial HARQ round given $b_2=\cdots=b_K=20$~bits and $\bar \gamma=20$~dB.}\label{EC_VL_LTAT:eps}
\end{figure}

Fig. \ref{LTAT_max_power_selection:eps} plots the optimal SE of VL-XP-HARQ and VL-IR-HARQ against the outage threshold $\varepsilon$ for different transmission powers ${P}$ and the maximum numbers of HARQ rounds $K$. It is observed that the optimal SE of VL-XP-HARQ is better than that of VL-IR-HARQ. Moreover, the optimal SE of both VL-XP-HARQ and VL-IR-HARQ converges to a performance bound under loose outage constraint, i.e., $\varepsilon  \to 1$. However, there is a non-negligible gap between the SE of VL-XP-HARQ and that of VL-IR-HARQ as the outage constraint is relaxed to 1. This is due to the fact that new information bits are added in retransmissions, which inevitably leads to the enhancement of the SE. 
\begin{figure}
\centering
\includegraphics[width=3.5in]{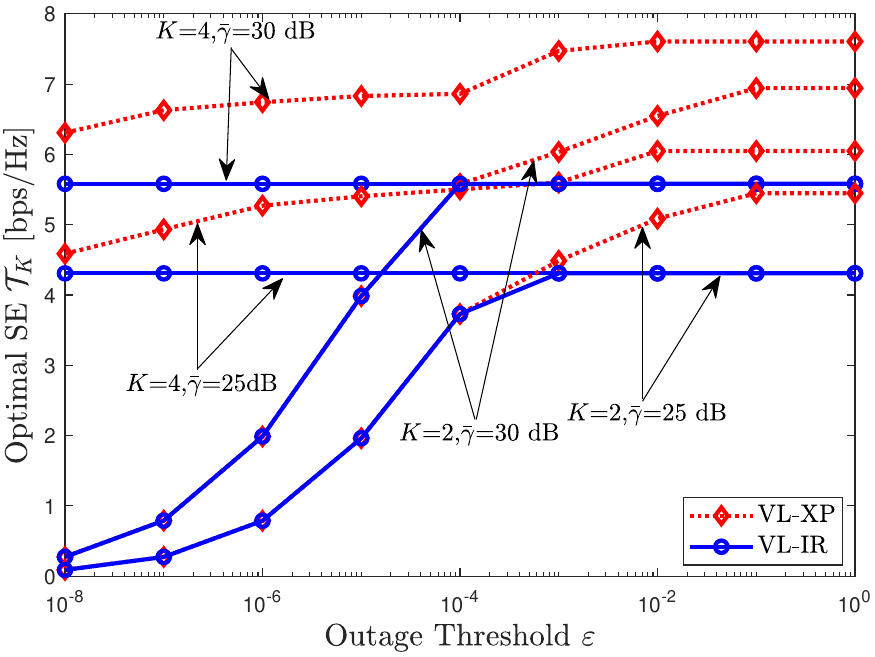}
\caption{The maximum SE versus the outage tolerance $\varepsilon$ by setting ${{\bf {N}}_2}=(100,200)$ symbols for $K=2$ and ${{\bf {N}}_4}=(100,200,201,202)$ symbols for $K=4$.}\label{LTAT_max_power_selection:eps}
\end{figure}

\subsection{EE}
To validate the analysis of the EE of VL-XP-HARQ, the simulated and analytical results are presented in Fig. \ref{EE:eps}. It is clear that the asymptotic results approach to the simulated ones under high SNR, which justifies the significance of the asymptotic analysis. In addition, it can be seen that the EE of VL-XP-HARQ coincides with that of VL-IR-HARQ in the high SNR regime. This is due to the fact that only one transmission is enough for one HARQ cycle at high SNR.
\begin{figure}
\centering
\includegraphics[width=3.5in]{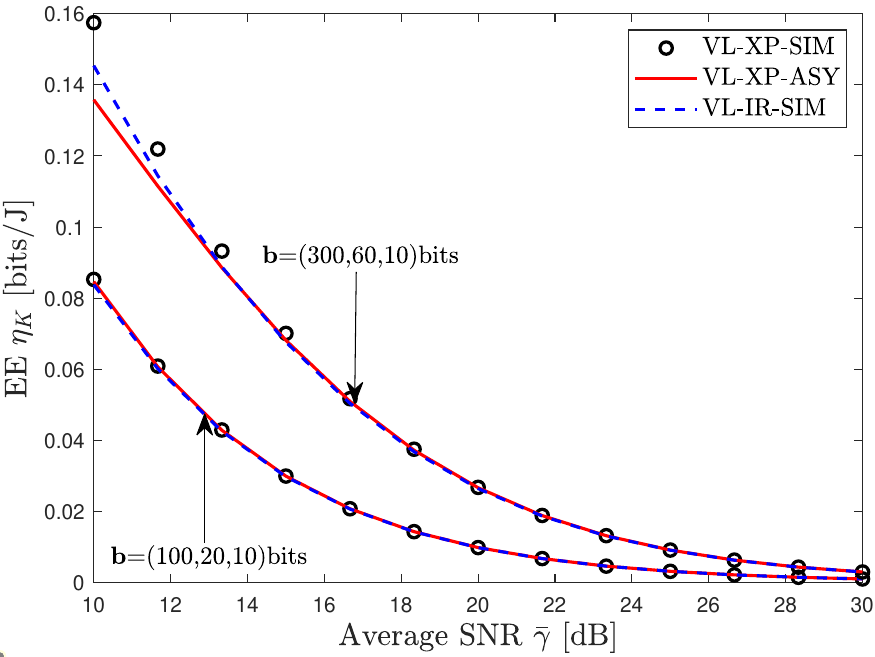}
\caption{The EE versus the average transmit SNR $\bar \gamma$ for $K=3$ and ${{\bf {N}}_3}=(100,200,201)$~symbols.}\label{EE:eps}
\end{figure}

To investigate the impact of $K$ on the EE, the EE of VL-XP-HARQ versus the number of information bits in the initial HARQ round, i.e., ${b_1}$ is shown in Fig. \ref{EE_bound_K:eps}, wherein the label ``Up-bound'' represents the upper bound of the EE given by \eqref{bound:EE}. It is found that the maximum EE increases with $K$. However, the EE is reduced if $b_1$ is sufficiently large. This is because the increasing the information bits will also degrades the outage performance, which consequently yields the decrease of the EE. In addition, in Fig. \ref{EE_bound_power:eps}, the EE is plotted against $b_1$ by considering different transmit SNR $\bar \gamma $. It can be seen that the maximum EE is improved with the reduction of $\bar \gamma$. This is consistent with our analytical results.

\begin{figure}
\centering
\includegraphics[width=3.5in]{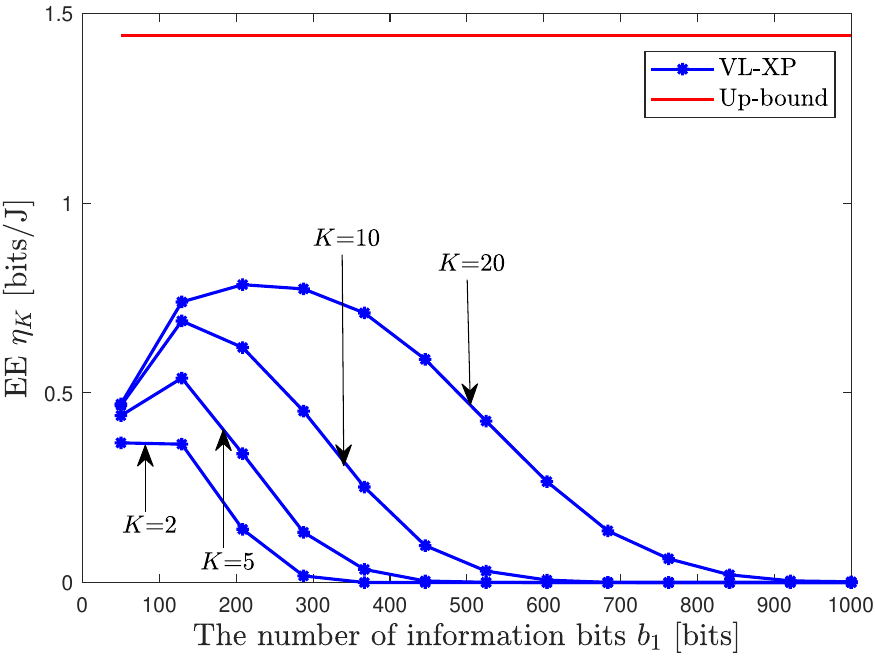}
\caption{The EE versus the number of information bits in initial HARQ round by setting $b_2=\cdots=b_K=20$~bits and $\bar \gamma=0$~dB.}\label{EE_bound_K:eps}
\end{figure}
\begin{figure}
\centering
\includegraphics[width=3.5in]{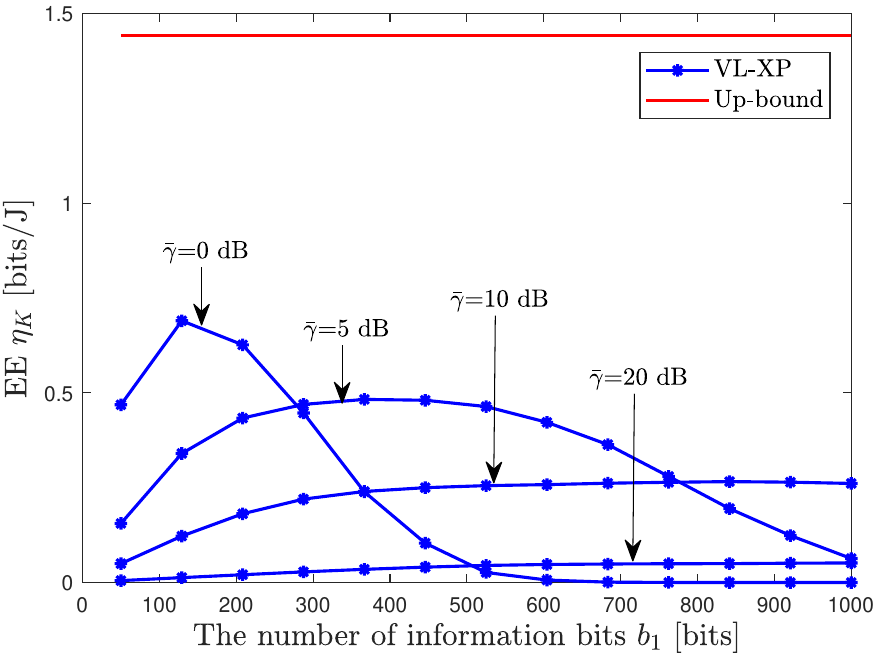}
\caption{The EE versus the number of information bits in initial HARQ round for $K=10$ and {$b_2=\cdots=b_K=20$~bits}.}\label{EE_bound_power:eps}
\end{figure}

In Fig. \ref{EE_max_power_selection:eps}, the optimal EE of both VL-XP-HARQ and VL-IR-HARQ versus the outage tolerance $\varepsilon $ is plotted by considering different $K$ and ${\bf{b}}$. Clearly, the optimal EE is reduced under a strict outage constraint. This due to the fact that the reliability and the efficiency are contradictory performance metrics. Moreover, it can seen from Fig. \ref{EE_max_power_selection:eps} that a larger $K$ can improve the EE of both VL-XP-HARQ and VL-IR-HARQ. 
However, a larger number of accumulated information bits ${\bf {b}}_K^\Sigma $ will reduce the EE. This is due to the reason that the optimal scheme prefers to save the power consumption, while large ${\bf {b}}_K^\Sigma $ will increase the outage probability, eventually result in the loss of the EE.

\begin{figure}
\centering
\includegraphics[width=3.5in]{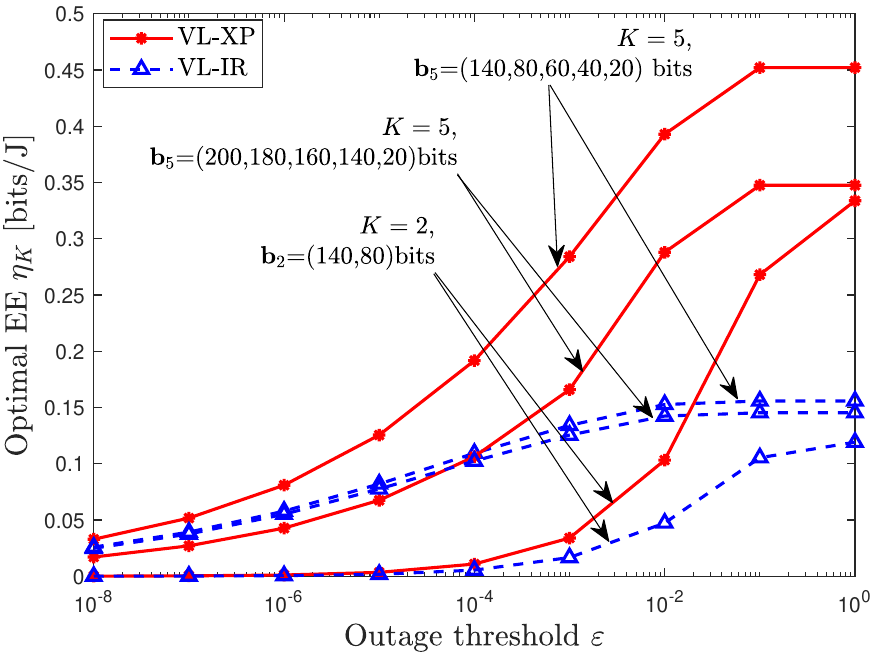}
\caption{The maximum EE versus the outage tolerance $\varepsilon$.}\label{EE_max_power_selection:eps}
\end{figure}
\section{Conclusion}
This paper has scrutinized the SE and the EE of VL-XP-HARQ. In particular, the closed-form expression of the SE has been obtained by renewal-reward theorem, and the relationship between SE and EC has been explored. In addition, the numbers of transmission bits in different HARQ rounds have been optimized to maximize the SE while maintaining an outage constraint. To solve this problem, Dinkelbach's transform has been adopted to convert the fractional programming problem into a subtraction form, which can be further decomposed into $K$ sub-problems through alternating optimization. On the basis of the monotonicity and convexity of the asymptotic outage probability, the successive convex approximation (SCA) has been invoked to relax each sub-problem to a convex problem. Aside from the SE, this paper has also examined the EE of VL-XP-HARQ. The upper bound of the EE has been proved to be attainable if the EC is achieved and the amount of power consumption tends to zero. Moreover, the similar method to the maximization of the SE has been applied to solve the maximization of the EE through optimal power allocation while ensuring the outage constraint. In the end, Monte Carlo simulations have been performed to validate our analytical results and VL-IR-HARQ has also been used to highlight the advantages of VL-XP-HARQ.




\appendices
\section{Proof of \eqref{eqn:out_for}}\label{app_out_for}
In order to derive the outage probability, it is imperative to obtain the achievable rate region of VL-XP-HARQ. According to Shannon information theory, if $K=1$, the outage event occurs if the mutual information is less than the number of information bits, i.e., $I({\bf x}_1;{\bf y}_1)<b_1$, where $I({\bf x}_k;{\bf y}_k)=N_k \log_2(1+\gamma_k)$. Thus the outage probability is given by $p_{out,1}=\Pr(I({\bf x}_1;{\bf y}_1)<b_1)$. To get a general expression of outage probability in the case of $K>1$, we next consider the case of $K=2$. The relevant results can be easily extended to the case with an arbitrary maximum number of HARQ rounds. In the meantime, we assume that Gaussian codebook and typical set decoding are leveraged. By following the similar proof as in \cite[Appendix A]{jabi2017adaptive}, the achievable rate region of VL-XP-HARQ can be obtained in the lemma given below.

\begin{lemma}\label{the:num_bits}
If the numbers of information bits delivered in two HARQ rounds meet the following three conditions
\begin{subequations}
    \begin{align}\label{eqn:inf_theory}
{b_1} + {b_2} &\le I\left( {{{\bf{x}}_1};{{{\bf{y}}}_{{1}}}} \right) + I\left( {{{\bf{x}}_2};{{{\bf{y}}}_{{2}}}} \right),\\
{b_2} &\le I\left( {{{\bf{x}}_2};{\bf{{y}}}_2} \right),  \\
    {b_1} &\le I\left( {{{\bf{x}}_1};{{{\bf{y}}}_{{1}}}} \right) + I\left( {{{\bf{x}}_2};{{{\bf{y}}}_{{2}}}} \right),
    \end{align}
\end{subequations}
there exists an HARQ-code $c^*$ with maximum probability of error $\epsilon\to 0$.
\end{lemma}

\begin{proof}
The successful decoding event occurs if the following sequences of symbols are jointly typical
\begin{equation}
    ({\left( {\Phi \left[ {{{\rm{m}}_1}} \right],{{{\bf{y}}}_1}} \right) \in A_\epsilon ^{ \star {N_1}}}),\, {\left( {\Phi \left[ {{{\rm{m}}_1},{{\rm{m}}_2}} \right],{{{\bf{y}}}_2}} \right) \in A_\epsilon ^{ \star {N_2}}}.
\end{equation}
Otherwise, the following error events take place:
\begin{subequations}
    \begin{align}
{{ E}_{0}} &= \left\{ {\left( {\Phi \left[ {{{\rm{m}}_1}} \right],{{\bf{y}}_1}} \right) \notin A_\epsilon ^{ \star {N_1}} \vee \left( {\Phi \left[ {{{\rm{m}}_1},{{\rm{m}}_2}} \right],{{\bf{y}}_2}} \right) \notin A_\epsilon ^{ \star {N_2}}} \right\},\\
{E_1} &= \left\{ \begin{array}{l}
\exists {{{\rm{\tilde m}}}_1} \ne {m_1},{{{\rm{\tilde m}}}_2} \ne {m_2}:\\
\left( {\Phi \left[ {{{{\rm{\tilde m}}}_1}} \right],{{\bf{y}}_1}} \right) \in A_\epsilon^{ \star {N_1}} \wedge \left( {\Phi \left[ {{{{\rm{\tilde m}}}_1},{{{\rm{\tilde m}}}_2}} \right],{{\bf{y}}_2}} \right) \in A_\epsilon^{ \star {N_2}},
\end{array} \right\}\\
{E_2} &= \left\{ {\begin{array}{*{20}{l}}
{\exists {{{\rm{\tilde m}}}_1} \ne {m_1}:}\\
{\left( {\Phi \left[ {{{{\rm{\tilde m}}}_1}} \right],{{\bf{y}}_1}} \right) \in A_\epsilon^{ \star {N_1}} \wedge  \left( {\Phi \left[ {{{{\rm{\tilde m}}}_1},{{\rm{m}}_2}} \right],{{\bf{y}}_2}} \right) \in A_\epsilon^{ \star {N_2}}}
\end{array}} \right\},\\
{E_3}& = \left\{ {\exists {{{\rm{\tilde m}}}_2} \ne {m_2}:\left( {\Phi \left[ {{{\rm{m}}_1},{{{\rm{\tilde m}}}_2}} \right],{{\bf{y}}_2}} \right) \in A_\epsilon^{ \star {N_2}}} \right\}.
    \end{align}
\end{subequations}
By using Packing lemma \cite[p.46]{el2011network}, the error probability is bounded as
\begin{align}\label{ineqn:dep}
{P_e} =& \Pr \left( {{E_0} \cup {E_1} \cup {E_2} \cup {E_3}} \right)\notag\\
 \le& \Pr \left( {{E_0}} \right) + \Pr \left( {{E_1}} \right) + \Pr \left( {{E_2}} \right) + \Pr \left( {{E_3}} \right)\notag\\
 \le& \epsilon  + {2^{3\left( {{N_1} + {N_2}} \right)\epsilon }}{2^{{b_1} + {b_2} - \left( {I\left( {{{\bf{x}}_1};{{\bf{y}}_1}} \right) + I\left( {{{\bf{x}}_2};{{\bf{y}}_2}} \right)} \right)}}\notag\\
 &+ {2^{3\left( {{N_1} + {N_2}} \right)\epsilon }}{2^{{b_1} - \left( {I\left( {{{\bf{x}}_1};{{\bf{y}}_1}} \right) + I\left( {{{\bf{x}}_2};{{\bf{y}}_2}} \right)} \right)}} + {2^{3{N_2}\epsilon }}{2^{{b_2} - I\left( {{{\bf{x}}_2};{{\bf{y}}_2}} \right)}}\notag\\
 \le& 4\epsilon,
\end{align}
where the second inequality holds by using \cite[Theorem 7.6.1]{cover1999elements}, and the last inequality holds if the three conditions in \eqref{eqn:inf_theory} are satisfied. Thus we complete the proof.
\end{proof}
With Lemma \ref{the:num_bits}, the decoding failure happens if and only if the receiver fails to decode the message in the first two HARQ rounds, i.e., the condition in \eqref{eqn:ach_rate} holds, where \eqref{eqn:ach_rate} is shown at the top of the next page, and herein ``ERR'' refers to error.
\begin{figure*}
\begin{align}\label{eqn:ach_rate}
{\rm ERR}_{\rm 1^{st}~Round} \wedge {\rm ERR}_{\rm 2^{nd}~Round}&=I({{\bf{x}}_1};{{\bf{y}}_1}) < {b_1} \wedge \overline{\left( {I({{\bf{x}}_2};{{\bf{y}}_2}) \ge {b_2} \vee (I({{\bf{x}}_1};{{\bf{y}}_1}) + I({{\bf{x}}_2};{{\bf{y}}_2}) \ge {b_1+b_2}} )\right)}\notag\\
&=I({{\bf{x}}_1};{{\bf{y}}_1}) < {b_1} \wedge \left( {I({{\bf{x}}_1};{{\bf{y}}_1}) + I({{\bf{x}}_2};{{\bf{y}}_2}) < {b_1+b_2}} \right)\notag\\
&={N_1}{\log _2}\left( {1 + {\gamma _1}} \right) < {b_1} \wedge \sum\nolimits_{k = 1}^2 {{N_k}{{\log }_2}\left( {1 + {\gamma _k}} \right)}  < \sum\nolimits_{k = 1}^2 {{b_k}}.
\end{align}
\hrulefill
\end{figure*}
It is easily found that the outage expression in \eqref{eqn:out_for} applies to $K=2$. By extending Lemma \ref{the:num_bits} for an arbitrary value of $K$, the generalized expression of the outage probability in \eqref{eqn:out_for} can be proved.



\section{Proof of theorem \ref{THEOREM ACCK2}}\label{acc_k2}
\subsubsection{${p_{out,1}}$}
If $k=1$, the outage probability ${p_{out,k}}$ can be obtained as
\begin{align}
    {p_{out,1}} &= \Pr \left( {{N_1}{{\log }_2}(1 + {\gamma _1}) < {b_1}} \right) \notag\\
    &= \Pr\left( {{{\left| {{h_1}} \right|}^2} < \frac{{{\sigma ^2}}}{{{P_1}}}\left( {{2^{{b_1}/{N_1}}} - 1} \right)} \right) \notag\\
    &=  1 - {e^{ - \left( {{2^{{b_1}/{N_1}}} - 1} \right){\sigma ^2}/{P_1}}},
\end{align}
where the last step holds because of the exponential distribution of ${\left| {{h_1}} \right|^2}$.
\subsubsection{${p_{out,2}}$}
According to \eqref{eqn:out_for}, the outage probability ${p_{out,2}}$ is expressed as
\begin{align}\label{eqn_ext_2}
&{p_{out,2}}= {{\mathbb E}_{{\gamma _2}}}\left\{ {\Pr \left( {{{\left| {{h_1}} \right|}^2} < \frac{{{\gamma _{\rm th}}{\sigma ^2}}}{{{P_1}}}\Big | {{\gamma _2}}} \right)} \right\}\notag\\
 &= {{\mathbb {E}}_{{\gamma _2}}}\left\{ {\left( {1 - {e^{ - \frac{{{\gamma _{\rm th}}{\sigma ^2}}}{{{P_1}}}}}} \right)u\left( {{2^{\frac{{{b_1} + {b_2} - {N_2}{{\log }_2}(1 + {\gamma _2})}}{{{N_1}}}}} - 1} \right)} \right\},
 \end{align}
  where ${\gamma _{\rm th}} = \min \left\{ {{2^{{b_1}/{N_1}}} - 1,{2^{\left( {{b_1} + {b_2} - {N_2}{{\log }_2}(1 + {\gamma _2})} \right)/{N_1}}} - 1} \right\}$, $u( \cdot )$ stands for the unit step function. Since 
  the received SNR ${\gamma_2}$ obeys exponential distribution with the probability density function (PDF) of ${f_{{\gamma _2}}}(x) = \exp ( - {\sigma ^2}x/{P_2}){\sigma ^2}/{P_2}$. Accordingly, \eqref{eqn_ext_2} can be expressed as
   \begin{multline}\label{eqn:outage_k2}
p_{out,2} = 
 {e^{ - \left( {{2^{{b_2}/{N_2}}} - 1} \right){\sigma ^2}/{P_2}}} - {e^{ - \left( {{2^{\left( {{b_1} + {b_2}} \right)/{N_2}}} - 1} \right){\sigma ^2}/{P_2}}}\\
 + \left( {1 - {e^{ - \left( {{2^{{b_1}/{N_1}}} - 1} \right){\sigma ^2}/{P_1}}}} \right)\left( {1 - {e^{ - \left( {{2^{{b_2}/{N_2}}} - 1} \right){\sigma ^2}/{P_2}}}} \right)\\
  - \underbrace {\frac{{{\sigma ^2}}}{{{P_2}}}{e^{\frac{{{\sigma ^2}}}{{{P_1}}}}}\int\limits_{{2^{\frac{{{b_2}}}{{{N_2}}}}} - 1}^{{2^{\frac{{{b_1} + {b_2}}}{{{N_2}}}}} - 1} {{e^{ - \frac{{{\sigma ^2}{2^{\frac{{{b_1} + {b_2} - {N_2}{{\log }_2}(1 + {\gamma _2})}}{{{N_1}}}}}}}{{{P_1}}}}}{e^{ - \frac{{{\sigma ^2}}}{{{P_2}}}{\gamma _2}}}d} {\gamma _2}}_{\varphi ({b_1},{b_2},{N_1},{N_2})}.
\end{multline}

By using the integral of Mellin-Branes type of the exponential function ${e^{ - u}} = \frac{1}{{2\pi {\mathop{\rm i}\nolimits} }}\int_{c - {\rm i}\infty }^{c + {\rm i}\infty } {\Gamma (s){u^{ - s}}ds} $ \cite[eq. 1.37]{mathai2009h}, 
$\varphi ({b_1},{b_2},{N_1},{N_2})$ can be expressed as
\begin{align}\label{eqn:outage_phi}
&\varphi \left( {{b_1},{b_2},{N_1},{N_2}} \right) \notag\\
&= {e^{\frac{{{\sigma ^2}}}{{{P_1}}} + \frac{{{\sigma ^2}}}{{{P_2}}}}}\frac{1}{{2\pi {\rm i}}}\int\limits_{c - {\rm i}\infty }^{c + {\rm i}\infty } {\Gamma (s){{\left( {\frac{{{\sigma ^2}{2^{\frac{{{b_1} + {b_2}}}{{{N_1}}}}}}}{{{P_1}}}{{\left( {\frac{{{\sigma ^2}}}{{{P_2}}}} \right)}^{\frac{{{N_2}}}{{{N_1}}}}}} \right)}^{ - s}}}\times \notag\\
&  \left( {\Gamma \left( {1 + \frac{{{N_2}}}{{{N_1}}}s,\frac{{{\sigma ^2}}}{{{P_2}}}{2^{\frac{{{b_2}}}{{{N_2}}}}}} \right) - \Gamma \left( {1 + \frac{{{N_2}}}{{{N_1}}}s,\frac{{{\sigma ^2}}}{{{P_2}}}{2^{\frac{{{b_1} + {b_2}}}{{{N_2}}}}}} \right)} \right)ds,
\end{align}
where $\Gamma(a)$ and $\Gamma(a,b)$ refer to the Gamma function and the upper incomplete Gamma function \cite[eq. 3.381.3]{gradshteyn1965table}. Furthermore, 
By using the definition of the generalized upper incomplete Fox's H function \cite{yilmaz2009productshifted}, \eqref{eqn:outage_phi} can be obtained in closed-form, with which \eqref{eqn:outage_k2} can be finally derived as \eqref{outage_k2}.

\section{Proof of Theorem \ref{THEOREM ASY}}\label{asy}
With \eqref{eqn:out_k_1_exa}, the asymptotic outage probability for $K=1$ can be easily obtained as \eqref{eqn:k_1_asy} by employing Taylor expansion of exponential function. Moreover, the asymptotic outage probabilities for $K=2$ and $K>2$ are individually derived in the following.
\subsection{The Case of $K=2$}
According to \eqref{outage_k2} in Theorem \ref{THEOREM ACCK2}, by applying residue theorem to generalized Fox's H function, \eqref{eqn:outage_phi} in the high SNR regime is rewritten as \eqref{eqn:phi_res}, as shown at the top of the next page, where
\begin{figure*}
    \begin{align}\label{eqn:phi_res}
\varphi \left( {{b_1},{b_2},{N_1},{N_2}} \right) = &{e^{\frac{{{\sigma ^2}}}{{{P_1}}} + \frac{{{\sigma ^2}}}{{{P_2}}}}}\sum\limits_{j =  - \infty }^0 {{\rm Res}\left\{ {\Gamma (s)\Gamma \left( {\frac{{{N_2}}}{{{N_1}}}s + 1,{2^{\frac{{{b_2}}}{{{N_2}}}}}\frac{{{\sigma ^2}}}{{{P_2}}}} \right){{\left( {{2^{\frac{{{b_1} + {b_2}}}{{{N_1}}}}}\frac{{{\sigma ^2}}}{{{P_1}}}{{\left( {\frac{{{\sigma ^2}}}{{{P_2}}}} \right)}^{\frac{{{N_2}}}{{{N_1}}}}}} \right)}^{ - s}},s=j} \right\}  }\notag\\
&-{{\rm Res}\left\{ {\Gamma (s)\Gamma \left( {\frac{{{N_2}}}{{{N_1}}}s + 1,{2^{\frac{{{b_1} + {b_2}}}{{{N_2}}}}}\frac{{{\sigma ^2}}}{{{P_2}}}} \right){{\left( {{2^{\frac{{{b_1} + {b_2}}}{{{N_1}}}}}\frac{{{\sigma ^2}}}{{{P_1}}}{{\left( {\frac{{{\sigma ^2}}}{{{P_2}}}} \right)}^{\frac{{{N_2}}}{{{N_1}}}}}} \right)}^{ - s}},s=j} \right\}},
\end{align}
\hrulefill
\end{figure*}
${\mathop{\rm Res}\nolimits} (f(s),{s=s_j})$ denotes the residue of $f(s)$ at pole ${s_j}$. By applying dominant term approximation 
to \eqref{eqn:phi_res}, only the residues at at ${s} = 0$ and ${s} = -1$ are kept. The asymptotic expression of $\varphi \left( {{b_1},{b_2},{N_1},{N_2}} \right)$ as $P_1/\sigma^2, P_2/\sigma^2\to \infty$ can be obtained as
\begin{align}\label{eqn:varphi_asy_res}
&\varphi \left( {{b_1},{b_2},{N_1},{N_2}} \right) \simeq {e^{\frac{{{\sigma ^2}}}{{{P_1}}} + \frac{{{\sigma ^2}}}{{{P_2}}}}} \times \notag\\
&\left( {\begin{array}{*{20}{l}}
{\Gamma \left( {1,\frac{{{\sigma ^2}}}{{{P_2}}}{2^{\frac{{{b_2}}}{{{N_2}}}}}} \right) - \Gamma \left( {1,\frac{{{\sigma ^2}}}{{{P_2}}}{2^{\frac{{{b_1} + {b_2}}}{{{N_2}}}}}} \right) - }\\
{\Gamma \left( {1 - \frac{{{N_2}}}{{{N_1}}},\frac{{{\sigma ^2}}}{{{P_2}}}{2^{\frac{{{b_2}}}{{{N_2}}}}}} \right)\left( {{2^{\frac{{{b_1} + {b_2}}}{{{N_1}}}}}\frac{{{\sigma ^2}}}{{{P_1}}}{{\left( {\frac{{{\sigma ^2}}}{{{P_2}}}} \right)}^{\frac{{{N_2}}}{{{N_1}}}}}} \right) + }\\
{\Gamma \left( {1 - \frac{{{N_2}}}{{{N_1}}},\frac{{{\sigma ^2}}}{{{P_2}}}{2^{\frac{{{b_1} + {b_2}}}{{{N_2}}}}}} \right)\left( {{2^{\frac{{{b_1} + {b_2}}}{{{N_1}}}}}\frac{{{\sigma ^2}}}{{{P_1}}}{{\left( {\frac{{{\sigma ^2}}}{{{P_2}}}} \right)}^{\frac{{{N_2}}}{{{N_1}}}}}} \right)}
\end{array}} \right),
\end{align}
wherein $\Gamma \left( {1, x } \right)=e^{-x}$ \cite[eq. 8.352.7]{gradshteyn1965table}. Furthermore, since the asymptotic expression of $\Gamma \left( {1 - \frac{{{N_2}}}{{{N_1}}},x} \right)$ as $x\to 0^+$ is given by \cite[eq.8.4.15, eq.8.7.3]{olver2010nist}
\begin{multline}\label{eqn:asy_incom_gamma}
    \Gamma \left( {1 - \frac{N_2}{N_1},x} \right) \simeq \\
    \left\{ {\begin{array}{*{20}{c}}
{\frac{{{N_1}}}{{{N_2} - {N_1}}}{x^{1 - {N_2}/{N_1}}},}&{\frac{{{N_2}}}{{{N_1}}} >1}\\
{{\frac{{{N_1}}}{{{N_2} - {N_1}}}{x^{1 - {N_2}/{N_1}}} + \Gamma \left( {1 - \frac{{{N_2}}}{{{N_1}}}} \right),}}&{\frac{{{N_2}}}{{{N_1}}}<1}
\end{array}} \right.,
\end{multline}
\eqref{eqn:varphi_asy_res} can be simplified as 
\begin{multline}\label{eqn:varphi_asy_fin}
\varphi \left( {{b_1},{b_2},{N_1},{N_2}} \right) \simeq{e^{\frac{{{\sigma ^2}}}{{{P_1}}} + \frac{{{\sigma ^2}}}{{{P_2}}}}}\left( {{e^{ - \frac{{{\sigma ^2}}}{{{P_2}}}{2^{\frac{{{b_2}}}{{{N_2}}}}}}} - {e^{ - \frac{{{\sigma ^2}}}{{{P_2}}}{2^{\frac{{{b_1} + {b_2}}}{{{N_2}}}}}}}} \right)\\
 + \frac{{{N_1}}}{{{N_1} - {N_2}}}{e^{\frac{{{\sigma ^2}}}{{{P_1}}} + \frac{{{\sigma ^2}}}{{{P_2}}}}}\frac{{{\sigma ^4}}}{{{P_1}{P_2}}}{2^{\frac{{{b_2}}}{{{N_2}}}}}\left( {{2^{\frac{{{b_1}}}{{{N_1}}}}} - {2^{\frac{{{b_1}}}{{{N_2}}}}}} \right).
\end{multline}
By substituting (\ref{eqn:varphi_asy_fin}) into (\ref{eqn:outage_k2}) along with Taylor expansion of the exponential function, ignoring the higher-order terms leads to \eqref{eqn:asy_2}.

\subsection{The Case of $K>2$}
On the basis of \eqref{eqn:out_for}, ${p_{out,K}}$ is rewritten as
\begin{align}\label{eqn:out_gen_mul}
p_{out,K} = \Pr \left( {{{(1 + {\gamma _1})}^{{N_1}}} < {2^{{b_1}}}, \cdots ,\prod\limits_{k = 1}^K {{{(1 + {\gamma _k})}^{{N_k}}}}  < {2^{b_K^\Sigma }}} \right)
\end{align}
By making the change of variable ${x_k} = \prod\nolimits_{l = 1}^k {{{(1 + {\gamma _l})}^{{N_l}}}}  $, \eqref{eqn:out_gen_mul} can be readily derived as
\begin{align}\label{eqn:outage_mul}
&{p_{out,K}} = \Pr \left( {{x_0} < {x_1} < {2^{{b_1}}}, \cdots ,{x_{K - 1}} < {x_K} < {2^{b_K^\Sigma }}} \right)\notag\\
& = \int\limits_{{x_0}}^{{2^{{b_1}}}} { \cdots \int\limits_{{x_{K - 1}}}^{{2^{b_K^\Sigma }}} {{f_{{x_1}, \cdots ,{x_K}}}({x_1}, \cdots ,{x_K})} } d{x_1} \cdots d{x_K},
\end{align}
where $x_0=1$. To proceed with the derivation, it is of necessity to obtain the joint PDF of ${\bf x} = (x_1, \cdots, {{x}_{K}})$. Since the joint PDF of ${\bs \gamma}=({{\gamma }_{1}},\cdots,{{\gamma }_{K}})$ is given by ${f_{{\bs \gamma}}}\left( {{\gamma _1}, \cdots ,{\gamma _K}} \right) =  \prod\nolimits_{k = 1}^K {{\sigma ^2}\exp \left( { - {\sigma ^2}{\gamma _k}/{P_k}} \right)/{P_k}} $, 
the joint PDF of $\bf x$ can be obtained by using Jacobian transform as ${{f}_{\bf x}}({{x}_{1}},\cdots ,{{x}_{K}})={{f}_{{\bs \gamma}}}({{\gamma }_{1}},\cdots ,{{\gamma }_{K}})|\det({\bf J})|$, where the matrix $\bf J$ reads as
\begin{align}\label{eqn:jacobi}
{\bf{J}} = \left( {\begin{array}{*{20}{c}}
{\frac{{{x_1}^{\frac{1}{{{N_1}}} - 1}}}{{{N_1}}}}&0& \cdots &0\\
{ - \frac{{{x_2}^{\frac{1}{{{N_2}}}}}}{{{N_2}{x_1}^{\frac{1}{{{N_2}}} + 1}}}}&{\frac{{{x_2}^{\frac{1}{{{N_2}}} - 1}}}{{{N_2}{x_1}^{\frac{1}{{{N_2}}}}}}}& \cdots &0\\
0& \vdots & \ddots & \vdots \\
0&0& \cdots &{\frac{{{x_K}^{\frac{1}{{{N_K}}} - 1}}}{{{N_K}{x_{K - 1}}^{\frac{1}{{{N_K}}}}}}}
\end{array}} \right).
\end{align}
Thus, we arrive at ${f_{\bf{x}}}({x_1}, \cdots ,{x_K}) = \prod\nolimits_{k = 1}^K {\frac{{{\sigma ^2}}}{{{P_k}}}{x_k}^{\frac{1}{{{N_k}}} - 1}\exp \left( { - \frac{{{\sigma ^2}}}{{{P_k}}}\left( {{{\left( {\frac{{{x_k}}}{{{x_{k - 1}}}}} \right)}^{\frac{1}{{{N_k}}}}} - 1} \right)} \right)/} {N_k}{x_{k - 1}}^{\frac{1}{{{N_k}}}} $. By substituting the joint PDF of ${\bf x}$ into \eqref{eqn:outage_mul}, $p_{out,K}$ can be obtained as
\begin{align}\label{eqn:out_toxk}
&p_{out,K} =
  \prod\limits_{k = 1}^K {\frac{{{\sigma ^2}}}{{{N_k}{P_k}}}{e^{\frac{{{\sigma ^2}}}{{{P_k}}}}}}\times \notag\\
 & \int\limits_{x_0}^{{2^{{b_1}}}} { \cdots \int\limits_{{x_{K - 1}}}^{{2^{b_K^\Sigma }}} {\prod\limits_{k = 1}^K {\frac{{{x_k}^{\frac{1}{{{N_k}}} - 1}}}{{{x_{k - 1}}^{\frac{1}{{{N_k}}}}}}{e^{ - \frac{{{\sigma ^2}}}{{{P_k}}}{{\left( {\frac{{{x_k}}}{{{x_{k - 1}}}}} \right)}^{\frac{1}{{{N_k}}}}}}}} } } d{x_1} \cdots d{x_K}.
\end{align}
Applying Taylor expansion of the exponential functions to \eqref{eqn:out_toxk} yields
\begin{align}\label{eqn:outage_MUL}
p_{out,K} &=
\prod\limits_{k = 1}^K {\frac{{{\sigma ^2}}}{{{N_k}{P_k}}}{e^{\frac{{{\sigma ^2}}}{{{P_k}}}}}\sum\limits_{{n_1}, \cdots ,{n_K}_{ = 0}}^\infty  {\left( {\prod\limits_{k = 1}^K {\frac{{{{\left( { - \frac{{{\sigma ^2}}}{{{P_k}}}} \right)}^{{n_k}}}}}{{{n_k}!}}} } \right)} } \notag\\
 &\times \int\limits_{x_0}^{{2^{{b_1}}}} { \cdots \int\limits_{{x_{K - 2}}}^{{2^{b_{K - 1}^\Sigma }}} {\prod\limits_{k = 1}^{K - 1} {{x_k}^{{\frac{{{n_k} + 1}}{{{N_k}}} - \frac{{{n_{k + 1}} + 1}}{{{N_{k + 1}}}} - 1}}} } } d{x_1} \cdots d{x_{K - 1}}\notag\\
& \times \left( {\frac{{{N_K}}}{{{n_K} + 1}}} \right)\left( {{2^{b_K^\Sigma \frac{{{n_K} + 1}}{{{N_K}}}}} - {x_{K - 1}}^{\frac{{{n_K} + 1}}{{{N_K}}}}} \right).
\end{align}
With the representation of the infinite summation in \eqref{eqn:outage_MUL}, the asymptotic outage probability at high SNR (i.e., ${P_1/\sigma^2} , \cdots  , {P_K/\sigma^2}\to \infty$) can be derived by removing the higher-order terms of $\prod\nolimits_{k=1}^{K}{{{\sigma^2/P_k}}}$. Clearly, the dominant term in \eqref{eqn:outage_MUL} is the one with ${{n}_{1}}=\cdots={{n}_{K}}=0$. Therefore, the asymptotic outage probability ${p_{out,K}}$ in the high SNR regime for $K>2$ is given by \eqref{eqn:outage_asy_mul}, as shown at the top of the next page.
\begin{figure*}
\begin{align}\label{eqn:outage_asy_mul}
{p_{out,K}} \simeq \prod\limits_{k = 1}^K {\frac{{{\sigma ^2}}}{{{N_k}{P_k}}}} \underbrace {\int\limits_{{x_0}}^{{2^{{b_1}}}} { \cdots \int\limits_{{x_{K - 2}}}^{{2^{b_{K - 1}^\Sigma }}} {\prod\limits_{k = 1}^{K - 1} {{x_k}^{\frac{1}{{{N_k}}} - \frac{1}{{{N_{k + 1}}}} - 1}} {N_K}\left( {{2^{\frac{{b_K^\Sigma }}{{{N_K}}}}} - {x_{K - 1}}^{\frac{1}{{{N_K}}}}} \right)} } d{x_1} \cdots d{x_{K - 1}}}_{{\hbar_{K,1}}({x_0})}.
\end{align}
\hrulefill
\end{figure*}


However, ${{\hbar_{K,1}}({x_0})}$ is a multi-fold integral that entails a high computational complexity on its numerical evaluation. Owing to the recursive form of ${\hbar_{K,k}}({x_{k - 1}})$ from \eqref{eqn:outage_asy_mul}, i.e.,
\begin{align}\label{eqn:h_proof}
{\hbar_{K,k}}({x_{k - 1}}) &= \int\nolimits_{{x_{k - 1}}}^{{2^{b_k^\Sigma }}} {{x_k}^{\frac{1}{{{N_k}}} - \frac{1}{{{N_{k + 1}}}} - 1}{h_{K,k + 1}}({x_k})} d{x_k},
\end{align}
${\hbar_{K,k}}({x_{k - 1}})$ can be obtained by induction. To start with the proof, we stipulate ${\hbar_{K,K}}({x_{K - 1}}) = {N_K}( {{2^{{{b_K^\Sigma }}/{{{N_K}}}}} - {x_{K - 1}}^{{1}/{{{N_K}}}}})$ on the basis of \eqref{eqn:outage_asy_mul}. Clearly, \eqref{eqn:out_h_def} holds for $k=K$.
To prove the recursive relationship in Theorem \ref{THEOREM ASY}, we assume that \eqref{eqn:out_h_def} holds for $k+1$, that is, ${\hbar_{K,k + 1}}({x_k})$ can be written as 
\begin{align}\label{eqn:h_k}
{\hbar_{K,k + 1}}({x_k})  &= {\left( { - 1} \right)^{K - k}}{x_k}^{\frac{1}{{{N_{k + 1}}}}}\prod\limits_{i = 1}^{K - k} {{N_{K - i + 1}}}   \notag\\
&+ \sum\limits_{i = 0}^{K - k - 1} {{c_{k + 1,i}}{x_k}^{\frac{1}{{{N_{k + 1}}}} - \frac{1}{{{N_{k + 1 + i}}}}}}.
\end{align}
Plugging \eqref{eqn:h_k} into \eqref{eqn:h_proof} results in 
\begin{align}\label{eqn:h_proof1}
{\hbar_{K,k}}({x_{k - 1}}) 
& = {\left( { - 1} \right)^{K - k}}{N_k}{2^{\frac{{b_k^\Sigma }}{{{N_k}}}}}\prod\limits_{i = 1}^{K - k} {{N_{K - i + 1}}}  \notag\\
&+ \sum\limits_{i = 1}^{K - k} {{c_{k + 1,i - 1}}\frac{{{N_k}{N_{k + i}}}}{{{N_{k + i}} - {N_k}}}{2^{\frac{{b_k^\Sigma }}{{{N_k}}} - \frac{{b_k^\Sigma }}{{{N_{k + i}}}}}}} \notag\\
& - \sum\limits_{i = 1}^{K - k} {{c_{k + 1,i - 1}}\frac{{{N_k}{N_{k + i}}}}{{{N_{k + i}} - {N_k}}}{x_{k - 1}}^{\frac{1}{{{N_k}}} - \frac{1}{{{N_{k + i}}}}}}  \notag\\
&+ {\left( { - 1} \right)^{K - k + 1}}{N_k}{x_{k - 1}}^{\frac{1}{{{N_k}}}}\prod\limits_{i = 1}^{K - k} {{N_{K - i + 1}}}.
\end{align}
By comparing \eqref{eqn:h_proof1} and \eqref{eqn:out_h_def}, the recursive relationship between the coefficients ${c_{k,i}}$ can be derived. Thus we complete the proof.

\section{Proof of Theorem \ref{THEOREM_III}}\label{outage_convex_proof}
By relaxing $b_r$ to a real number, we use the function $f_K(b_1,\cdots,b_K)$ to define $\prod\nolimits_{k=1}^K N_k^{-1}{{\hbar_{K,1}}({x_0})}$ for notational simplicity, i.e,
\begin{multline}\label{f_z}
{f_K}\left( {{b_1}, \cdots ,{b_K}} \right) =\\
\prod\limits_{k=1}^K \frac{1}{N_k}\int\limits_{{x_0}}^{{2^{{b_k}}}} { \cdots \int\limits_{{x_{K - 2}}}^{{2^{\sum\nolimits_{k = 1}^{K - 1} {{b_k}} }}} {\prod\limits_{k = 1}^{K - 1} {{x_k}^{\frac{1}{{{N_k}}} - \frac{1}{{{N_{k + 1}}}} - 1}} } } \\
\times {N_K}\left( {{2^{\frac{{\sum\nolimits_{k = 1}^{{K}} {{b_k}} }}{{{N_K}}}}} - {x_{K - 1}}^{\frac{1}{{{N_K}}}}} \right)d{x_1} \cdots d{x_{K - 1}}.
\end{multline}
According to \eqref{eqn:outage_asy_mul}, the asymptotic outage probability can be expressed as
\begin{equation}\label{eqn:p_out_tilde}
    {\tilde p_{out,K}}(\mathbf b) = \prod\limits_{k = 1}^K {\frac{{{\sigma ^2}}}{{{P_k}}}} {f_K}\left( {{b_1}, \cdots ,{b_K}} \right).
\end{equation}
To prove Theorem \ref{THEOREM_III}, it suffices to show the increasing monotonicity and the convexity of $f_K(b_1,\cdots,b_K)$ with respect to $b_r$ given $b_1,\cdots,b_{r-1},b_{r+1},\cdots,b_K$. Henceforth, it amounts to determining that the first-order and second-order partial derivatives of $ {f_K}\left( {{b_1}, \cdots ,{b_K}} \right)$ are larger than or equal to zero.
To this end, on the basis of \eqref{eqn:out_for}, ${p_{out,K}}$ is rewritten as 
\begin{multline}\label{out_convex}
{p_{out,K}} =
\Pr \left( {\ln \left( {1 + \frac{{{P_1}{{\left| {{h_1}} \right|}^2}}}{{{\sigma ^2}}}} \right) < {b_1}\ln 2} \right.\\
\left. {, \cdots ,\ln \prod\limits_{k = 1}^K {\left( {1 + \frac{{{P_k}{{\left| {{h_k}} \right|}^2}}}{{{\sigma ^2}}}} \right)}  < \sum\limits_{k = 1}^K {{b_k}} \ln 2} \right).
\end{multline}
Since the random variables ${\left| {{h_1}} \right|^2},\cdots,{\left| {{h_K}} \right|^2}$ follow exponential distributions, \eqref{out_convex} can be explicitly expressed as \eqref{out_convex1}, as shown at the top of the following page.
\begin{figure*}
\begin{align}\label{out_convex1}
{p_{out,K}} = \int\nolimits_{\ln {{\left( {1 + \frac{{{P_1}{x_1}}}{{{\sigma ^2}}}} \right)}^{{N_1}}} < {b_1}\ln 2} { \cdots \int\nolimits_{\ln \prod\limits_{k = 1}^K {{{\left( {1 + \frac{{{P_k}{x_k}}}{{{\sigma ^2}}}} \right)}^{{N_k}}} < \sum\limits_{k = 1}^K {{b_k}} \ln 2} } {\prod\limits_{k = 1}^K {{e^{ - {x_k}}}} d{x_1} \cdots d{x_K}} }
\end{align}
\hrulefill
\end{figure*}
From \eqref{out_convex1}, the domain of integration of $x_1,\cdots,x_K$ should satisfy the constraints $\ln \prod\nolimits_{l = 1}^k {{{\left( {1 + \frac{{{P_l}{x_l}}}{{{\sigma ^2}}}} \right)}^{{N_l}}}}  < \ln 2\sum\nolimits_{l = 1}^k {{b_l}}$. Clearly, in the high SNR, as ${P_1/\sigma^2} , \cdots  , {P_K/\sigma^2}\to \infty$, $x_1,\cdots,x_K$ approach to zero to meet these constraints.
Hence, the integrand $e^{-x_k}$ in \eqref{out_convex1} tends to 1.  \eqref{out_convex1} is thus asymptotically equal to
\begin{multline}\label{out_convex2}
  p_{out,K} \simeq \\
  \int\limits_0^\infty  { \cdots \int\limits_0^\infty  {\prod\limits_{k = 1}^K {u\left( {\sum\limits_{l = 1}^k {{b_l}} \ln 2 - \ln \prod\limits_{l = 1}^k {{{\left( {1 + \frac{{{P_l}}}{{{\sigma ^2}}}{x_l}} \right)}^{{N_l}}}} } \right)} } } \\
  d{x_1} \cdots d{x_K},
\end{multline}
where $u( \cdot )$ stands for the unit step function. By making use of the inverse Laplace transform of $u( x )$, i.e., $u(x) = \frac{1}{{2\pi {\mathop{\rm i}\nolimits} }}\int_{c - {\mathop{\rm i}\nolimits} \infty }^{c + {\mathop{\rm i}\nolimits} \infty } {\frac{1}{s}{e^{sx}}} ds$, and after some algebraic manipulations, \eqref{out_convex2} can be rewritten as
\begin{align}\label{out_convex3}
    p_{out,K} &\simeq \prod\limits_{k = 1}^K {\frac{{{\sigma ^2}}}{{{P_k}}}}{\left( {\frac{1}{{2\pi {\rm i}}}} \right)^K} \notag\\
    &\times \int\limits_{{c_1} - {\rm{i}}\infty }^{{c_1} + {\rm{i}}\infty } { \cdots \int\limits_{{c_K} - {\rm{i}}\infty }^{{c_K} + {\rm{i}}\infty } {\prod\limits_{k = 1}^K {\frac{1}{{{s_k}\left( {{N_k}\sum\nolimits_{l = k}^K {{s_l}}  - 1} \right)}}} } } \notag\\
    &\times {2^{{s_k}\sum\limits_{l = 1}^k {{b_l}} }}d{s_1} \cdots d{s_K},
\end{align}
where ${\rm i}=\sqrt{-1}$. By identifying \eqref{out_convex3} with \eqref{eqn:p_out_tilde} and \eqref{out_convex1}, $f_K(b_1,\cdots,b_K)$ can also be written as \eqref{out_convex4}, as shown at the top of the next page.
\begin{figure*}
\begin{align}\label{out_convex4}
{f_K}\left( {{b_1}, \cdots ,{b_K}} \right) &= {\left( {\frac{1}{{2\pi {\rm{i}}}}} \right)^K}\int\nolimits_{{c_1} - {\rm{i}}\infty }^{{c_1} + {\rm{i}}\infty } { \cdots \int\nolimits_{{c_K} - {\rm{i}}\infty }^{{c_K} + {\rm{i}}\infty } {\prod\limits_{k = 1}^K {\frac{1}{{{s_k}\left( {{N_k}\sum\nolimits_{l = k}^K {{s_l}}  - 1} \right)}}{2^{{s_k}\sum\nolimits_{l = 1}^k {{b_l}} }}} } } d{s_1} \cdots d{s_K}\notag\\
& = \int\nolimits_{\ln {{\left( {1 + {x_1}} \right)}^{{N_1}}} < {b_1}\ln 2} { \cdots \int\nolimits_{\ln \prod\nolimits_{k = 1}^K {{{\left( {1 + {x_k}} \right)}^{{N_k}}} < \sum\nolimits_{k = 1}^K {{b_k}} \ln 2} } {d{x_1} \cdots d{x_K}} } .
\end{align}
\hrulefill
\end{figure*}
From \eqref{out_convex4}, it is readily found that ${f_K}\left( {{b_1}, \cdots ,{b_K}} \right)$ is an increasing function of ${b_{r}}$ by fixing ${b_1}, \cdots ,{b_{r - 1}},{b_{b + 1}}, \cdots ,{b_K}$, because the domain of the integration is enlarged as ${b_{r}}$ increases. Accordingly, it is proved that the first partial derivative of ${f_K}\left( {{b_1}, \cdots ,{b_K}} \right)$ w.r.t. ${b_{r}}$ is larger than or equal to zero, i.e., $\frac{{\partial {f_K}\left( {{b_1}, \cdots ,{b_K}} \right)}}{{\partial {b_r}}} \ge 0$, wherein
\begin{multline}\label{out_convex5}
\frac{{\partial {f_K}\left( {{b_1}, \cdots ,{b_K}} \right)}}{{\partial {b_r}}} = \ln 2{\left( {\frac{1}{{2\pi {\rm{i}}}}} \right)^K} \int\limits_{{c_1} - {\rm{i}}\infty }^{{c_1} + {\rm{i}}\infty } { \cdots \int\limits_{{c_K} - {\rm{i}}\infty }^{{c_K} + {\rm{i}}\infty } {\sum\limits_{l = r}^K {{s_l}} } }   \\
 \times\prod\limits_{k = 1}^K {\frac{{{2^{{s_k}\sum\limits_{l = 1}^k {{b_l}} }}}}{{{s_k}\left( {{N_k}\sum\limits_{l = k}^K {{s_l}}  - 1} \right)}}} d{s_1} \cdots d{s_K},
\end{multline}
Then taking the second-order partial derivative of ${f_K}\left( {{b_1}, \cdots ,{b_K}} \right)$ with respect to ${b_{r}}$ gives
\begin{multline}\label{out_convex6}
\frac{{{\partial ^2}{f_K}\left( {{b_1}, \cdots ,{b_K}} \right)}}{{\partial b_r^2}} = {\left( {\ln 2} \right)^2}{\left( {\frac{1}{{2\pi {\rm{i}}}}} \right)^K}\\
 \times \int_{{c_1} - {\rm{i}}\infty }^{{c_1} + {\rm{i}}\infty } { \cdots \int_{{c_K} - {\rm{i}}\infty }^{{c_K} + {\rm{i}}\infty } {{{\left( {\sum\limits_{l = r}^K {{s_l}} } \right)}^2}} } \\
\times \prod\limits_{k = 1}^K {\frac{{{2^{{s_k}\sum\limits_{l = 1}^k {{b_l}} }}}}{{{s_k}\left( {{N_k}\sum\limits_{l = k}^K {{s_l}}  - 1} \right)}}} d{s_1} \cdots d{s_K}.
\end{multline}
By using the following identity $$\sum\limits_{l = r}^K {{s_l}}  = \frac{{{N_r}\sum\nolimits_{l = r}^K {{s_l}}  - 1}}{{{N_r}}} + \frac{1}{{{N_r}}},$$
\eqref{out_convex6} can be rewritten as
\begin{align}\label{out_convex6}
\frac{{{\partial ^2}{f_K}\left( {{b_1}, \cdots ,{b_K}} \right)}}{{\partial b_r^2}} = &\frac{{\ln 2}}{{{N_r}}}\frac{{\partial {f_{K/r}}\left( {{b_1}, \cdots ,{b_K}} \right)}}{{\partial {b_r}}}\notag\\
&+ \frac{{\ln 2}}{{{N_r}}}\frac{{\partial {f_K}\left( {{b_1}, \cdots ,{b_K}} \right)}}{{\partial {b_r}}},
\end{align}
where ${{\partial {f_{K/r}}\left( {{z_1}, \cdots ,{z_K}} \right)}}/{{\partial {z_r}}}$ is given by
\begin{multline}\label{out_convex7}
\frac{{\partial {f_{K/r}}\left( {{z_1}, \cdots ,{z_K}} \right)}}{{\partial {z_r}}} = \ln 2{\left( {\frac{1}{{2\pi {\rm{i}}}}} \right)^K} \\
 \times \int_{{c_1} - {\rm{i}}\infty }^{{c_1} + {\rm{i}}\infty } { \cdots \int_{{c_K} - {\rm{i}}\infty }^{{c_K} + {\rm{i}}\infty } {\sum\nolimits_{l = r}^K {{s_l}} \prod\limits_{k = 1,k \ne r}^K {\frac{1}{{{N_k}\sum\nolimits_{l = k}^K {{s_l}}  - 1}}} } }  \\
 \times \prod\limits_{k = 1}^K {\frac{{{2^{{s_k}\sum\nolimits_{l = 1}^k {{b_l}} }}}}{{{s_k}}}} d{s_1} \cdots d{s_K},
\end{multline}
where ${f_{K/r}}\left( {{b_1}, \cdots ,{b_K}} \right)$ is defined by
\begin{equation}
    {{\tilde p}_{out,K/r}}({\bf{b}}) = \prod\limits_{k = 1,k \ne r}^K {\frac{{{\sigma ^2}}}{{{P_k}}}} {f_{K/r}}\left( {{b_1}, \cdots ,{b_K}} \right),
\end{equation}
and ${{\tilde p}_{out,K/r}}({\bf{b}})$ is the asymptotic outage probability given no power consumption in the $r$-th HARQ round, i.e., $P_r=0$. Hence, similarly to \eqref{out_convex4}, ${f_{K/r}}\left( {{b_1}, \cdots ,{b_K}} \right)$ is explicitly given by \eqref{eqn:fKsubr}, as shown at the top of the next page.
\begin{figure*}
    \begin{align}\label{eqn:fKsubr}
{f_{K/r}}\left( {{b_1}, \cdots ,{b_K}} \right) =& \int_{\ln {{\left( {1 + {x_1}} \right)}^{{N_1}}} < {b_1}\ln 2} { \cdots \int_{\ln \prod\nolimits_{k = 1,k \ne r}^K {{{\left( {1 + {x_k}} \right)}^{{N_k}}}}  < \sum\nolimits_{k = 1}^K {{b_k}} \ln 2} {{e^{ - {x_r}}}d{x_1} \cdots d{x_K}} } \notag\\
 =& {\left( {\frac{1}{{2\pi {\rm{i}}}}} \right)^K}\int_{{c_1} - {\rm{i}}\infty }^{{c_1} + {\rm{i}}\infty } { \cdots \int_{{c_K} - {\rm{i}}\infty }^{{c_K} + {\rm{i}}\infty } {\prod\limits_{k = 1,k \ne r}^K {\frac{1}{{{N_k}\sum\nolimits_{l = k}^K {{s_l}}  - 1}}} } } \prod\limits_{k = 1}^K {\frac{{{2^{{s_k}\sum\nolimits_{l = 1}^k {{b_l}} }}}}{{{s_k}}}} d{s_1} \cdots d{s_K}.
    \end{align}
    \hrulefill
\end{figure*}
Clearly, ${{\partial {f_{K/r}}\left( {{b_1}, \cdots ,{b_K}} \right)}}/{{\partial {b_r}}} \ge 0$ in analogous to \eqref{out_convex5}. Substituting this result into \eqref{out_convex6} easily demonstrates ${{\partial^2 {f_{K}}\left( {{b_1}, \cdots ,{b_K}} \right)}}/{{\partial {b_r^2}}} \ge 0$. Thus we finally prove the increasing monotonicity and convexity of the asymptotic outage probability.

\bibliographystyle{IEEEtran}
\bibliography{VL_XP_manuscript}
\end{document}